\documentclass[10pt,journal]{IEEEtran}
\ifCLASSOPTIONcompsoc
  \usepackage[nocompress]{cite}
\else
  \usepackage{cite}
\fi
\ifCLASSINFOpdf
   \usepackage[pdftex]{graphicx}
\else
   \usepackage[dvips]{graphicx}
\fi
\usepackage{amsmath}
\allowdisplaybreaks 

\usepackage{enumitem}
\usepackage{mathtools}
\usepackage{amssymb}
\usepackage{amsthm}
\usepackage{amsfonts}
\usepackage{bm}
\usepackage{dsfont}
\usepackage{xcolor}
\usepackage{algpseudocode}
\usepackage{algorithm}
\usepackage{multirow}
\usepackage{caption}%
\usepackage{subcaption}
\usepackage{cuted}

\newtheorem{theorem}{Theorem}
\newtheorem{definition}{Definition}
\newtheorem{example}{Example}
\newtheorem{lemma}{Lemma}

\newtheorem{remark}{Remark}

\DeclareMathOperator*{\argmin}{\arg\!\min}
\DeclareMathOperator*{\argmax}{\arg\!\max}

\newcommand{\mc}{\mathcal}
\newcommand{\mb}{\mathbb}

\begin{document}
%
\title{
Limited-Trust in Diffusion of Competing Alternatives over Social Networks
}
%
%
%
%

\author{Vincent~Leon,~
        S.~Rasoul~Etesami,~\IEEEmembership{Member,~IEEE,}
        and~Rakesh~Nagi,~\IEEEmembership{Senior~Member,~IEEE}
\IEEEcompsocitemizethanks{\IEEEcompsocthanksitem Vincent Leon, S. Rasoul Etesami, and Rakesh Nagi are with Department of Industrial and Enterprise Systems Engineering and Coordinated Science Laboratory, University of Illinois at Urbana-Champaign, Urbana, IL, 61801, USA. Email: (leon18, etesami1, nagi)@illinois.edu.
}
\thanks{
This material is based upon work supported by the Air Force Office of Scientific Research under award number FA9550-23-1-0107 and the NSF CAREER Award under grant number EPCN-1944403. A shorter version of this paper has been presented at the 2022 61st IEEE Conference on Decision and Control (CDC), Cancun, Mexico \cite{leon2022diffusion}.}
}

\IEEEtitleabstractindextext{%
\begin{abstract}
We consider the diffusion of two alternatives in social networks using a game-theoretic approach. Each individual plays a coordination game with its neighbors repeatedly and decides which to adopt. As products are used in conjunction with others and through repeated interactions, individuals are more interested in their long-term benefits and tend to show trust to others to maximize their long-term utility by choosing a suboptimal option with respect to instantaneous payoff. To capture such trust behavior, we deploy limited-trust equilibrium (LTE) in diffusion process. We analyze the convergence of emerging dynamics to equilibrium points using mean-field approximation and study the equilibrium state and the convergence rate of diffusion using absorption probability and expected absorption time of a reduced-size absorbing Markov chain. We also show that the diffusion model on LTE under the best-response strategy can be converted to the well-known linear threshold model. Simulation results show that when agents behave trustworthy, their long-term utility will increase significantly compared to the case when they are solely self-interested. Moreover, the Markov chain analysis provides a good estimate of convergence properties over random networks.
\end{abstract}

}

\maketitle

\IEEEdisplaynontitleabstractindextext

%
\IEEEpeerreviewmaketitle

\section{Introduction}\label{sec:introduction}

Imagine that a few competing products exist in the market, and the manufacturers would like to know whether their product will become popular and dominate the market. Suppose that each alternative has certain group of users, and the utility that each alternative brings about is available to all. Individuals interact with others who may have adopted different alternatives and choose one alternative strategically.
Through interactions, one of the alternatives may gradually spread in the network as more individuals adopt it because they can benefit more by switching to using the new product. Alternatively, it is possible that one product eventually dies out because it brings about less utility to consumers or has poor compatibility with other competing alternatives.

An example of strategic choice among competing alternatives is smartphone operation systems. Android and iOS have attracted consumers from other platforms and become dominant in the past decade. By contrast, Symbian has lost its predominance in market share, and the service for another competing platform Windows Phone has been discontinued. Other similar scenarios include strategic choice among competing brands of game consoles or network service providers. In the former, one may take into account which brand of game console its friends have when purchasing a game console. If it prefers sharing games and accounts with its friends to save the cost of game purchases, it can buy the same brand as the majority. It can buy a different brand if it prefers a wider variety of games across different platforms. In the example of the network service provider, one may consider joining the same network with its family or friends when subscribing to a mobile service, as signing up on the same network and setting up a family plan can reduce the cost per subscriber. As a result, individuals' decision-making may ultimately affect the market share of game consoles or network subscriptions.

Products like cell phones, cameras, and computers are likely to be used over a long term. As individuals frequently interact with others and update the choice of their product in conjunction with others, the long-term benefit obtained from consuming a product becomes more important than the short-term benefit that one can derive in isolation. That motivates individuals to give up some of their immediate gains in order to benefit in the long term. Therefore, from a managerial perspective, a good understanding of the consumers' long-term behavior can help companies better promote their products over the socioeconomic network. 

The above scenarios are some examples of diffusion of competing alternatives over socioeconomic networks, and one can consider many other applications such as epidemics spread in biological networks \cite{pare2020modeling}, information spread in computer networks \cite{khanafer2014information}, and belief propagation in social networks \cite{etesami2021open}. Diffusion of competing alternatives (e.g., product, virus, opinion, etc.) has been extensively studied in the past literature, where broadly speaking, the goal is to propagate certain types of products or behaviors in a desired way through the network \cite{foster1990stochastic,young1993evolution,kandori1993,kempe2003maximizing,kempe2005influential,young2006diffusion,young2011dynamics,montanari2010,kreindler2014,fazli2014nonprogressive,gursoy2018influence,chan2020influence,arieli2020speed,ghayoori2021seed}. In general, there are two main approaches to analyzing the diffusion: i) epidemic-based modeling, in which individuals adopt an alternative based on the intensity of being exposed to it \cite{kempe2003maximizing,acemoglu2011diffusion}, and ii) game-theoretic modeling, in which individuals make a decision about adopting an alternative in order to maximize their utility functions \cite{montanari2010,etesami2016complexity}. In this work, we follow the second approach to analyze the behavior of diffusion dynamics under a different equilibrium concept.

In this paper, we focus on a game-theoretic model for diffusion dynamics. More precisely, we consider diffusion of a superior alternative based on the notion of \emph{limited-trust equilibrium} (LTE) \cite{murray2021} over social networks. 
LTE aims to explain individual's rationality regarding short-term and long-term utilities, which accounts for an individual's behavior of deviating from the apparent best strategy for one-shot game and investing in future by showing a certain degree of generosity towards others (see Section \ref{subsec:prelim-lte}). In each iteration of diffusion, individuals play a coordination game with their neighbors and show a certain degree of trust when making a decision. It is worth noting that one can view LTE in the form of Nash equilibrium (NE) using a non-smooth transformation of payoff function. However, analyzing diffusion dynamics under such a transformation of payoff function is difficult. Instead, we analyze the LTE diffusion process directly based on the payoff matrix associated with competing alternatives, which has been adopted in the literature of diffusion dynamics (see Section \ref{subsec:literature}). The main contributions of this paper are as follows:
\begin{itemize}
    \item Developing a diffusion model based on the notion of LTE;
    \item Analyzing the convergence rate and long-term behavior of LTE diffusion dynamics using Markov chains and mean-field approximations; 
    \item Establishing a connection between the LTE diffusion model and the standard linear-threshold model.
\end{itemize}
Simulation results show that the LTE diffusion dynamics result in desirable convergence properties compared to the conventional NE diffusion dynamics, which is based on the assumption that individuals choose the apparent best response of a single-step diffusion (see Section \ref{subsec:diffusion-ne}). When the superior alternative is the risk-averse option, both LTE and NE diffusion dynamics converge to the preferred option, and LTE diffusion dynamics converges faster than NE dynamics. When the superior alternative and the risk-averse alternative do not coincide, the NE diffusion dynamics converges to the risk-averse and undesirable option, whereas the LTE diffusion dynamics converges to the superior alternative or converges to the risk-averse option at a slower rate than the NE dynamics. The average utility by following LTE diffusion dynamics increases drastically compared to the case of NE diffusion dynamics, which reveals that prevalence of the superior alternative is associated with the individuals' desire for better long-term utility provided that they behave trustworthy. We also justify the theoretical bounds obtained from our Markov chain analysis using extensive simulation results. The analytical results from the Markov chain in terms of \emph{probability of domination of the superior alternative} and \emph{expected time to equilibrium} match the simulation results reasonably well, especially when the underlying structure of the Markov chain is 1-dimensional. 


\subsection{Related Work} \label{subsec:literature}

There has been a long literature on the evolution of diffusion dynamics over social networks \cite{acemoglu2011diffusion,kempe2003maximizing,etesami2016complexity}. It captures the essence of ``viral marketing'' which takes advantage of the power of ``word-of-mouth'' among the interactions of individuals in order to promote a product, a new technology, a belief, etc. 
The epidemic-based models are based on the exposure of individuals to the influence. The two most notable epidemic-based diffusion models are independent cascade (IC) model and linear threshold (LT) model \cite{kempe2003maximizing}. In the IC model, each newly activated individual can activate each of its inactive neighbors with some probability; in the LT model, each agent has a threshold, and an inactive agent will become active if the sum of weights from its activated neighbors exceeds the threshold.

The LT and IC models are closely related to the \emph{influence maximization} problem in which the goal is to choose a fixed-size set of initial adopters in a social network to maximize the spread of influence. For the LT and IC models, the seed selection problem can be formulated as a combinatorial optimization problem \cite{kempe2003maximizing,kempe2005influential}.  For these models, it was shown in \cite{kempe2003maximizing} that the objective function of the influence maximization problem is a submodular function of the initial seed set, and a hill-climbing greedy algorithm can be used to get a $(1-1/e)$-approximation of the optimal seed set. Subsequent studies directly related to the original LT and IC models are \cite{kempe2005influential, mossel2010submodularity}, which extended the IC and LT models to a more general setting and generalized the results about submodularity. Some variants of the LT model include the non-progressive LT model in which an individual can switch between active and inactive states \cite{fazli2014nonprogressive,chan2020influence} and the deterministic LT model where threshold values are input to the model \cite{gursoy2018influence}. An extension of the IC model is the multi-cascaded diffusion model, where multiple types of cascade spread over the social network \cite{ghayoori2021seed}.

Apart from the diffusion model itself, seed selection is a dedicated topic and the central problem for various diffusion models. Early work on seed selection is centrality-based heuristics, including degree and distance centrality heuristics \cite{wasserman1994,borgatti2006,hinz2011} that associate the influence of a node to its degree or its distance from other nodes.  \cite{samadi2016subjective} designed Bayesian evidence-based parallel cascade model to analyze the diffusion dynamics of two competing products and to study the optimal seed selection strategy based on such dynamics. The parallel cascade model assumes that an individual can be inactive, positive active, or negative active. \cite{samadi2017temporal} investigates the temporal aspects of the diffusion process based on the parallel cascade model, which helps decision-makers set up time horizons for short-term and long-term goals. 
A generalization of the seed selection problem is the seed activation scheduling \cite{chierichetti2014,samadi2018seed}, in which the decision-makers are allowed to activate the seeds in a sequence as diffusion processes rather than activating the seeds all at once.

On the other hand, most game-theoretic diffusion models are based on utility maximization under the concept of Nash equilibrium (NE). 
\cite{kandori1993,young2006diffusion} study dynamics of the diffusion process in which an individual interacts with its neighbors by playing a coordination game and makes a rational decision by choosing the product that maximizes its instantaneous utility.
\cite{kandori1993} investigated the equilibrium emerging from the long-term behavior of the noisy best-response dynamics and concluded that the long-term equilibrium coincides with the risk-dominant equilibrium introduced by \cite{harsanyi1988}. The work \cite{young2006diffusion} further characterized the long-run behavior and convergence properties of the noisy best response dynamics in terms of the geometry of underlying network. More recent results in that line of research include \cite{montanari2010,kreindler2014,arieli2020speed}, which provided bounds on the expected convergence time of the noisy best response dynamics that are either dependent or independent of the network structure.

The above literature on game-theoretic models relies on individuals being myopic: individuals choose the option that maximizes their instantaneous utility. On the other hand, there is abundant literature showing that players may deviate from the apparent best strategy or NE strategy of single-shot game in dynamic setting, which leads to the appearance of altruism and cooperation. Earlier work \cite{kreps1982reputation} shows this phenomenon mathematically when players have incomplete information, and \cite{kreps1982rational} shows that such a setting can lead to cooperative behavior of the two players in finitely repeated prisoner's dilemma. \cite{basar1998,basar2018handbook,chung2016game,leon2021bandit} study dynamic games in which players aim to maximize a certain notion of cumulative utility over a time horizon. In such settings, a player generally has to consider both instantaneous and future utilities when making a decision, hence resulting in a decision strategy that may deviate from the apparent best strategy for its instantaneous utility. 

There is abundant literature in the area of economics that studies the phenomenon where players deviate from the apparent NE strategy in strategic games.
Various behavioral models and experiments have been proposed and designed. Some well-known behavioral models include level-$k$ \cite{nagel1995unraveling,stahl1994experimental,stahl1995players,costa-gomes2001cognition,costa-gomes2006cognition}, cognitive hierarchy \cite{camerer2004cognitive}, and generalized cognitive hierarchy models \cite{chong2016generalized}, which assume that players have different levels of thinking ability ranging from random to rational. The strategy for level-0 players is specified, and higher-level players best respond to lower-level players. Other behavioral models include quantal response equilibrium \cite{mckelvey1995quantal} and cursed equilibrium \cite{eyster2005cursed} models. 
For experiments, the authors of \cite{palfrey1996altruism} designed an experiment to measure the effect of altruism, reputation, and noise in voluntary contribution games and reached a conclusion that the apparent altruistic behavior was principally due to random variation, whereas altruism and reputation-building played a minor role in that.
In \cite{carrillo2009compromise}, an experiment of a compromise game was designed, and the experiment results deviated from the NE when players had private information. Several above-mentioned cognitive-behavioral models are applied to the experiment results.
Finally, in \cite{ali2021adverse}, the authors designed an experiment to study the difference in behavior towards adverse and advantageous selections when players were informed asymmetrically. The experiment results indicate that players account more for adverse than advantageous selection, and such a difference can be connected to learning from feedback and self-confirming equilibria.


A dedicated line of research studies trust and reputation in dynamic games. In \cite{palfrey1996altruism,berg1995trust,englewarnick2006learning}, the authors conducted experiments on repeated games and provided theoretical analyses of trust to account for the behavior that deviates from the NE of one-off games. Various notions of trust have been proposed to account for the behavior of choosing a suboptimal strategy regarding one-stage utility. The introduction of trust as a parameter implicitly induces another game on top of the original one, and players play rationally in the induced game.
In \cite{ledyard1995public}, the authors propose \emph{$\alpha$-altruism}, a notion of trust suitable for simultaneous one-shot games. On this notion, each individual aims to maximize its perceived utility which accounts for social welfare as a part of its utility. In an earlier work \cite{palfrey1996altruism}, an $\alpha$-altruism-like notion was also adopted in the hypothesis and experiment analysis.
\cite{murray2021} has recently introduced the concept of limited-trust equilibrium (LTE) to explain individuals' trust behavior that may lead to higher social welfare \cite{murray2022network}.

%


\subsection{Organization and Notation}
The paper is organized as follows. In Section \ref{sec:preliminaries}, we provide some preliminary results on LTE and conventional NE diffusion dynamics. In Section \ref{sec:model}, we formally introduce the LTE diffusion model. In Section \ref{sec:mean-field-approx}, we transform the LTE diffusion model to a reduced-size Markov chain and characterize its convergence time and behavior. In Section \ref{sec:threshold-model}, we make a connection between the LTE diffusion model and the notable linear-threshold diffusion model and show their equivalence under certain regimes. Simulation results are provided in Section \ref{sec:simulation}. We conclude the paper in Section \ref{sec:conclusion}.
Additionally, we provide an e-companion \cite{leon2022limitedtrust-arxiv} for numerical examples, proof of theorem and additional simulation results. 

\noindent
{\bf Notation:} We use bold fonts of lower-case letters to denote vectors (e.g., $\bm{x}$) and subscript indices to denote their components (e.g., $x_i$). We use bold fonts of upper-case letters to denote matrices (e.g., $\bm{P}$) and subscript indices to denote their entries (e.g., $P_{ij}$). For any positive integer $k$, we let $[k] \triangleq \{1, \cdots, k\}$.

\section{Preliminaries}\label{sec:preliminaries}

\subsection{Limited-Trust Equilibrium}\label{subsec:prelim-lte}

Under the conventional NE, each individual behaves selfishly. However, in many cases, individuals make choices that do not appear to benefit them in the short term, showing some degree of altruism to others with the hope of a similar return from them. For example, people show concern for others and for their well-being. If each individual considers only its own immediate benefit, no one will care for others, and no one will get support because one cannot get benefited instantly by spending time listening to and encouraging others. 
The authors of \cite{murray2021} have introduced a new concept of equilibrium named \emph{limited-trust equilibrium} (LTE) to account for such phenomena. On the notion of LTE, a player aims to maximize its long-term utility by trusting other players and giving up a certain amount of utility so that others gain considerably more than what it loses, hoping that other players will return that favor in future. LTE provides a new explanation to such phenomena where individuals invest a certain amount of short-term utility for future returns.



Formally, in a social network $G([n], E)$, each player $i \in [n]$ has a hard \emph{trust limit} (or \emph{trust level}) $\delta_i> 0$. On the concept of LTE, player $i$ chooses the strategy that maximizes the social welfare among all the strategies whose utilities are at most $\delta_i$ less than the maximum utility that player $i$ can obtain. Consider a finite $n$-player game. Let $\Sigma_i$ denote the strategy set of player $i$ and $\sigma_i \in \Sigma_i$ be a strategy. Denote by $\sigma=(\sigma_1, \cdots, \sigma_n)$ the strategy of all players. Given a player $i$, denote by $\sigma_{-i}=(\sigma_1, \cdots, \sigma_{i-1}, \sigma_{i+1}, \cdots, \sigma_n)$ the strategy of the players other than $i$, and let $u_i(\sigma_i, \sigma_{-i})$ and $u(\sigma_i, \sigma_{-i})=\sum_{j=1}^n u_j(\sigma_j, \sigma_{-j})$ be the utility of player $i$ and the social welfare respectively. The set \(\sigma_i^G(\sigma_{-i}) \triangleq \argmax_{\sigma_i\in \Sigma_i} u_i(\sigma_i, \sigma_{-i})\) is defined as the \emph{greedy best response} of player $i$ given $\sigma_{-i}$. Let $\sigma_i^G \in \sigma_i^G(\sigma_{-i})$, and denote the trust limits of all the players by $\bm{\delta}=(\delta_1, \cdots, \delta_n)$, where $\delta_i > 0, \forall i\in [n]$. 

\begin{definition}
A strategy profile $\sigma=(\sigma_1, \cdots, \sigma_n)$ is a limited-trust equilibrium (LTE) if and only if \(u_i(\sigma_i^G,\sigma_{-i})-u_i(\sigma_i,\sigma_{-i})\leq \delta_i\) and \(u(\sigma_i,\sigma_{-i})\geq u(\sigma_i',\sigma_{-i})\) for all $i\in [n]$ and any \(\sigma_i'\in \{\sigma_i' \in \Sigma_i: u_i(\sigma_i^G,\sigma_{-i})-u_i(\sigma_i',\sigma_{-i})\leq \delta_i\}\).
\end{definition}

Therefore, given fixed strategy of other players $\sigma_{-i}$, the \emph{limited-trust best response} of player $i$, denoted by $\sigma_i^*(\sigma_{-i})$, can be obtained by solving the following program: 
\begin{align}
    \sigma_i^*(\sigma_{-i}) = & \argmax_{\sigma_i \in \Sigma_i} u(\sigma_i, \sigma_{-i}), \tag{P1} \label{eq:lt-best-resp-LP} \\
    \text{s.t.} \quad & u_i(\sigma_i^G,\sigma_{-i}) - u_i(\sigma_i,\sigma_{-i}) \leq \delta_i. \notag
\end{align}

Like NE, one can show that LTE exists in every $n$-player finite game with positive trust limits. In particular, computation of LTE in general is PPAD-hard \cite{murray2021}.

\subsection{Game-theoretic Model for Diffusion under NE} \label{subsec:diffusion-ne}

A social network of $n$ players is represented by a weighted undirected graph $G([n], E)$ with edge weights $\{w_{ij}\}_{(i,j)\in E}$. Each vertex represents one player in the social network, and each edge represents the interaction between two players. For an edge $(i,j)\in E$, the interaction between $i$ and $j$ is symmetric. The weights of edges describe the probability, strength or importance of the interaction. Each player $i$ has two choices: $A$ and $B$, i.e., $\Sigma_i=\{A,B\}$. The state of player $i$ is denoted by $x_i\in \{A,B\}$, and the state of the process is denoted by $\bm{x}=(x_1, \cdots, x_n)\in \{A, B\}^n$. 

The utility for player $i$ of choosing $A$ and $B$ is composed of two components, an individual component, denoted by $v_i(x_i)$, representing the player's private preference for $A$ or $B$ irrespective of other players' choices, and a social component resulted from the externalities due to other players' choices. The social component for $i$ is expressed as $\sum_{j\in N(i)} w_{ij}v(x_i,x_j)$, which is the sum of utilities from the interaction with each of $i$'s neighbors $N(i)$. Here, $v(x_i, x_j)$ is the payoff function of a two-person \emph{coordination game} in which each player has two pure strategies $A$ and $B$. The payoff matrix of the coordination game is shown in Table \ref{tab:coord-game-payoff-matrix}.
It is assumed that conformity is always better than non-conformity. That is, in the interaction with each of its neighbors, the player gets higher payoff when it chooses the same product as the neighbor than when it chooses the different product. This implies that $a > d$ and $b > c$. Therefore, the (total) utility to agent $i$ in state $\bm{x}$ equals \(u_i(\bm{x}) = \sum_{j\in N(i)} w_{ij}v(x_i,x_j) + v_i(x_i)\). 

\begin{example}
Assume that the players are users of iPhone and Android phone in a social network. 
Each player has its own preference for iPhone or Android phone, which is reflected in $v_i(x_i)$ for all $i\in [n]$. The externality due to the interaction between two users is embodied in $v(x_i,x_j)$. Most likely, each type of phone works better with the same type (e.g., file transfer), and hence the payoff through interaction is greater when two users are in conformity.
\end{example}


\begin{table}[h]
    \captionsetup{font=small}
    \caption{Coordination game payoff matrix with $a>d$ and $b>c$.}
    \centering
    \begin{tabular}{|c||c|c|}
        \hline
          & $A$ & $B$  \\ \hline \hline
        $A$ & $a$, $a$ & $c$, $d$ \\ \hline
        $B$ & $d$, $c$ & $b$, $b$ \\ \hline
    \end{tabular}
    \label{tab:coord-game-payoff-matrix}
\end{table}

In fact, one can show that the above $n$-person game is a \emph{potential game} with the potential function \cite{young2006diffusion}:
\begin{equation}\label{eq:NE-diffusion-potential}
    \Phi(\bm{x}) = (a-d) w_{AA}(\bm{x}) + (b-c) w_{BB}(\bm{x}) + \sum_{i=1}^n v_i(x_i),
\end{equation}
where $w_{AA}(\bm{x})$ is the sum of the weights on all edges $(i,j)$ such that $x_i=x_j=A$, and $w_{BB}(\bm{x})$ is defined in a similar way. Therefore, one can obtain a pure-strategy Nash equilibrium of this game by maximizing the potential function \eqref{eq:NE-diffusion-potential}. As maximizing \eqref{eq:NE-diffusion-potential} in general could be a difficult task, a popular method for obtaining its NE points is to analyze the long-term behavior of certain natural diffusion dynamics that may arise due to repeated interactions of the players over the social network.

The diffusion dynamics is a continuous-time process in which each player is assumed to have an independent Poisson clock that rings once per unit time in expectation. A player updates its strategy every time its Poisson clock rings. This can be naturally extended to a discrete-time process in which, at each time instant, one player is randomly selected to update its strategy. When a player updates its strategy, instead of deterministically playing its best response strategy, the player is assumed to choose a strategy with probability described by the following logit function:

\noindent
\underline{NE Diffusion Model}
\begin{equation} \label{eq:nash-diffusion-logit-prob}
    P(x_i=A | \bm{x}_{-i}) = \frac{e^{\beta u_i(A, \bm{x}_{-i})}}{e^{\beta u_i(A, \bm{x}_{-i})}+e^{\beta u_i(B, \bm{x}_{-i})}},
\end{equation}
where $\beta$ is a parameter measuring the sensitivity of the agent to payoff differences. As $\beta \to \infty$, the stochastic logit dynamics by Eq. \eqref{eq:nash-diffusion-logit-prob} degenerates to deterministic best response dynamics. Hence, on the top of rationality, Eq. \eqref{eq:nash-diffusion-logit-prob} introduces randomness and stochasticity into the diffusion process. 

As is shown in \cite{young2006diffusion,foster1990stochastic}, following the stochastic diffusion process by Eq. \eqref{eq:nash-diffusion-logit-prob} induces a long-run distribution over the game state $\boldsymbol{x}$, which is concentrated on the states that maximize the potential function by Eq. \eqref{eq:NE-diffusion-potential}, i.e., the pure NE points. In particular, if $v_i(A)= v_i(B)$ for all $i\in [n]$, i.e., everyone is indifferent between $A$ and $B$, the potential function \eqref{eq:NE-diffusion-potential} reduces to
\begin{equation}\label{eq:NE-diffusion-potential-social-only}
    \Phi(\bm{x}) = (a-d) w_{AA}(\bm{x}) + (b-c) w_{BB}(\bm{x}),
\end{equation}
and the diffusion process by Eq. \eqref{eq:nash-diffusion-logit-prob} will converge to the risk-dominant equilibrium \cite{harsanyi1988,kandori1993,young2006diffusion}. If $a-d > b-c$, all-$A$ state is the risk-dominant equilibrium, and $A$ is the risk-dominant strategy. If $a-d < b-c$, all-$B$ state and $B$ are the risk-dominant equilibrium and risk-dominant strategy, respectively. The risk-dominant equilibrium and strategy reflect the risk-averse behavior in face of uncertainty.

\subsection{Linear Threshold Model}

One of the notable models to study the influence spread and epidemic is the linear threshold (LT) model introduced in \cite{kempe2003maximizing}.
The social network is represented by a directed graph $G([n], E)$. Each player, represented by a vertex, has two states: active and inactive. There are two types of LT model: \emph{progressive} and \emph{non-progressive} models. In the progressive model, once a player becomes active, it will remain active forever and cannot be inactive again; in the non-progressive model, each individual can switch back and forth between active and inactive states. The progressive model is adopted for the conventional LT model. Each player $i\in [n]$ is assigned a threshold $\theta_i$ drawn uniformly at random from $[0,1]$. Each player $i$ is influenced by its out-neighbor $j \in N^+(i)$ according to a weight $b_{ij}$ such that $\sum_{j \in N^+(i)} b_{ij} \leq 1$. At the beginning, an initial set of active players is chosen. In the diffusion process, at step $t=1,2,\ldots$, each player $i \in [n]$ examines the states of its neighbors and becomes active if $\sum_{\substack{j \in N^+(i) \\ j \text{ active}}} b_{ij} \geq \theta_i$. Given $\theta_i$ for all $i\in [n]$ and initial set of active players, the diffusion progresses deterministically over the network until no more activation takes place.

\section{Diffusion Dynamics under Limited-Trust Equilibrium}\label{sec:model}


We consider the diffusion of a superior alternative on an undirected social network $G([n],E)$ with two competing alternatives $A$ and $B$. As before, each vertex represents one player, and each edge represents an interconnection between two players. For $i\in [n]$, we let $N(i)$ be the set of neighbors of player $i$ and $d(i)$ be the degree of vertex $i$. We assume that $G$ is a simple graph such that it does not contain self-loops or parallel edges. This is a realistic assumption because individuals obtain information about a product from the people they know and obtain utility through their interaction with them. 

The diffusion process takes place over a fixed time horizon $T$ and is modeled as a discrete-time process. During the process, each player can choose between two strategies: $A$ and $B$. At time instant $t\in [T]$, one player is randomly picked to update its strategy, in which case that player plays a coordination game with each of its neighbors. The payoff matrix of the coordination game is the same as the one for the diffusion model under NE as shown in Table \ref{tab:coord-game-payoff-matrix}, where without loss of generality, we may assume $a, b, c, d \geq 0$. As in the diffusion model under NE, we assume conformity is always better than non-conformity, i.e., $a > d$ and $b > c$. Furthermore, we assume that $A$ is the superior alternative that brings about more utility, i.e., $a > b$. We can express the state of all players at time $t$ using an $n$-dimensional vector \(\bm{x}(t)\in\{A, B\}^n\). Note that this discrete-time process can be naturally extended to a continuous-time process in which each player has an independent Poisson clock with unit rate, and the player updates its strategy every time its Poisson clock rings. 

During the diffusion process, each player makes decision regarding the competing alternatives by considering its short-term and long-term utilities. To realize this, each player takes into account both its utility and its neighbors' utilities when it makes a decision and evaluates the gain and loss of both factors. When the cost of selecting some product which may not be the best one for the player in terms of utility but maximizes the social welfare can be afforded, the player will choose that product to benefit its neighbors and hope that its neighbors will return the favor to it in the same way in future. 

More precisely, let $x_i \in \{A, B\}$ denote the state of player $i$, and denote by $v(x_i, x_j)$ the payoff of $i$ in the coordination game between $i$ and its neighbor $j$. The total utility of player $i$ by playing $x_i$ equals $u_i(x_i, \bm{x}_{-i}) = \sum_{j \in N(i)} v(x_i, x_j)$, and the social welfare under state $\bm{x}$, denoted by $u(\bm{x})$, is given by
\[u(\bm{x}) = \sum_{i\in V} u_i(\bm{x}) = \sum_{(i,j)\in E} [v(x_i, x_j) + v(x_j, x_i)].\] 
Given trust levels \(\bm{\delta} = (\delta_1, \cdots, \delta_n) > 0\), let
\begin{align*}
    S_i(\bm{x}_{-i}) = & \Big\{x_i \in \{A, B\}: \\
    & \quad \max_{x_i' \in \{A, B\}} u_i(x_i', \bm{x}_{-i}) - u_i(x_i, \bm{x}_{-i}) \leq \delta_i\Big\}
\end{align*}
denote the set of strategies whose utilities are at least the utility of the best response minus player $i$'s trust level. This is equivalent to finding the strategy set satisfying the constraint in \ref{eq:lt-best-resp-LP}.
Since the strategy set of player $i$ contains only two elements, the strategies in $\argmax_{x_i \in \{A, B\}} u_i(x_i, \bm{x}_{-i})$ are in $S_i(\bm{x}_{-i})$, and whether the strategies in $\argmin_{x_i \in \{A, B\}} u_i(x_i, \bm{x}_{-i})$ are in $S_i(\bm{x}_{-i})$ depends on the difference between $\max_{x_i \in \{A, B\}} u_i(x_i, \bm{x}_{-i})$ and $\min_{x_i \in \{A, B\}} u_i(x_i, \bm{x}_{-i})$. If \(\min_{x_i \in \{A, B\}} u_i(x_i, \bm{x}_{-i}) \geq \max_{x_i \in \{A, B\}} u_i(x_i, \bm{x}_{-i}) - \delta_i\), then both $A$ and $B$ are in $S_i(\bm{x}_{-i})$ and satisfy the limited-trust constraint in \ref{eq:lt-best-resp-LP}. As a result, both strategies are under player $i$'s consideration, and the player will choose whichever maximizes the social welfare. On the contrary, if \(\min_{x_i \in \{A, B\}} u_i(x_i, \bm{x}_{-i}) < \max_{x_i \in \{A, B\}} u_i(x_i, \bm{x}_{-i}) - \delta_i\), $\argmax_{x_i \in \{A, B\}} u_i(x_i, \bm{x}_{-i})$ contains only one strategy, and this strategy is the only choice for the player because the player will lose too much utility if it chooses the strategy not in $\argmax_{x_i \in \{A, B\}} u_i(x_i, \bm{x}_{-i})$ regardless of the social welfare. 
By introducing randomness into the model in a similar manner and adopting the logit dynamics described by \eqref{eq:nash-diffusion-logit-prob} in the NE diffusion model, we can formally introduce the LTE diffusion model. Let us define 
\begin{align*}\nonumber
    W_i(\bm{x}_{-i}) &= \argmax_{x_i\in \{A, B\}} u(x_i,\bm{x}_{-i}), \\
    U_i(\bm{x}_{-i}) &= \argmax_{x_i\in \{A, B\}} u_i(x_i,\bm{x}_{-i}), \\
    u^*_i(\bm{x}_{-i}) &= \max_{x_i\in\{A, B\}} u_i(x_i,\bm{x}_{-i}), \\
    u'_i(\bm{x}_{-i}) &= \min_{x_i\in\{A, B\}} u_i(x_i,\bm{x}_{-i}).
\end{align*}
Then, the LTE diffusion dynamics can be described by

\medskip
\noindent
\underline{LTE Diffusion Model}
\begin{enumerate}[leftmargin=*]
    \item If \(W_i(\bm{x}_{-i}) \cap U_i(\bm{x}_{-i}) \neq \varnothing\), 
    \begin{equation} \label{eq:lte-diffusion-1-prob}
        P(x_i \in W_i(\bm{x}_{-i}) \cap U_i(\bm{x}_{-i})| \bm{x}_{-i}) = 1;
    \end{equation}
    
    \item If \(W_i(\bm{x}_{-i}) \cap U_i(\bm{x}_{-i}) = \varnothing\),
    \begin{enumerate}
        \item \(u'_i(\bm{x}_{-i}) \geq u^*_i(\bm{x}_{-i}) - \delta_i\),
        \begin{equation} \label{eq:lte-diffusion-2a-prob}
            P(x_i = A | \bm{x}_{-i}) = \frac{e^{\beta' u(A, \bm{x}_{-i})}}{e^{\beta' u(A, \bm{x}_{-i})} + e^{\beta' u(B, \bm{x}_{-i})}};
        \end{equation}
        
        \item \(u'_i(\bm{x}_{-i}) < u^*_i(\bm{x}_{-i}) - \delta_i\),
        \begin{equation} \label{eq:lte-diffusion-2b-prob}
            P(x_i = A|\bm{x}_{-i}) = \frac{e^{\beta u_i(A, \bm{x}_{-i})}}{e^{\beta u_i(A, \bm{x}_{-i})} + e^{\beta u_i(B, \bm{x}_{-i})}}.
        \end{equation}
    \end{enumerate}
\end{enumerate}
\noindent
where $\beta'$ and $\beta$ are the parameters measuring the sensitivity of players to the difference between social welfare and utility, respectively. A player who maximizes the social welfare benefits its neighbors and hence aims to gain more utility in the long run as its neighbors return the favor to it when it is their turn. In this regard, social welfare embodies a player's long-term utility. Therefore, $\beta'$ reflects the sensitivity of a player to its long-term utility, whereas $\beta$ reflects the sensitivity to the short-term utility as in the NE model. When the social welfare and utility maximizers coincide (Case 1), the player chooses the maximizing product deterministically because it is beneficial to both short-term and long-term utilities. When the social welfare and utility maximizers differ (Case 2), the player will compare the gain and the loss. As in the NE diffusion model, instead of choosing the best-response strategy deterministically, the players are allowed to ``make a mistake'' and choose the ``worse'' alternative with a small probability. The probability function for decision-making is softened by a logistic function. A player may attach different importance to long-term and short-term utilities, reflecting different sensitivity to social welfare and utility.  
Simulation results are provided in Section \ref{sec:simulation} to compare the diffusion process under NE and LTE dynamics. Simulations show that, on average, players gain more utility when they behave trustworthy, which results in a rapid spread of the superior alternative over the network. 

\subsection{Relationship with Potential Games for Diffusion} \label{subsec:potential}

As we have seen in Section \ref{subsec:diffusion-ne}, the diffusion under NE dynamics is a potential game with the potential function as follows \cite{young2006diffusion}:
\begin{equation}\label{eq:NE-potential-unwei}
    \Phi(\bm{x}) = (a-d) E_{AA}(\bm{x}) + (b-c) E_{BB}(\bm{x}),
\end{equation}
where $E_{AA}(\bm{x})$ is the number of all edges $(i, j)$ such that $x_i = x_j = A$, and likewise $E_{BB}(\bm{x})$ is defined. (Note that Eq. \eqref{eq:NE-potential-unwei} is the unweighted version of Eq. \eqref{eq:NE-diffusion-potential-social-only}.) Under NE dynamics, the incentive function, the function which each player aims to maximize, is the player's own utility. It can be shown that under NE dynamics, the change of the player's incentive function matches the change in the value of Eq. \eqref{eq:NE-potential-unwei}. Under LTE dynamics, it is more complicated to analyze the diffusion process as the potential function may not exist. 
When each player has zero trust limit, it is the same as the NE dynamics, and the diffusion game is a potential game with the potential function by Eq. \eqref{eq:NE-potential-unwei}. 
When each player has infinite trust limit, the incentive function of each player is social welfare, and the diffusion process is another potential game with the potential function as follows: 
\begin{equation} \label{eq:SW-potential-unwei}
    \Psi(\bm{x}) = (2a-c-d) E_{AA}(\bm{x}) + (2b-c-d) E_{BB}(\bm{x}).
\end{equation}
It can also be shown that any change in social welfare matches the change in the function value by Eq. \eqref{eq:SW-potential-unwei}. 
When the player's trust limit is between zero and infinity, the potential function may not exist. If both $A$ and $B$ satisfy the constraint in \ref{eq:lt-best-resp-LP} (see Section II-A) for some player, the player's incentive function is social welfare; otherwise, if only one of $A$ and $B$ (the greedy best response) satisfies the constraint, the player's incentive function is utility. In each state, different players may have different incentive functions due to different positions in the network, different choices of alternatives, and heterogeneous values of trust limits. As a result of the non-uniqueness of incentive function, the potential function may not exist. Therefore, the diffusion process under LTE dynamics is not a potential game in general. 

\begin{figure}[ht]
    \captionsetup{font=small}
    \centering
    \includegraphics[trim=.4in .2in .4in .4in, clip, width=0.4\textwidth]{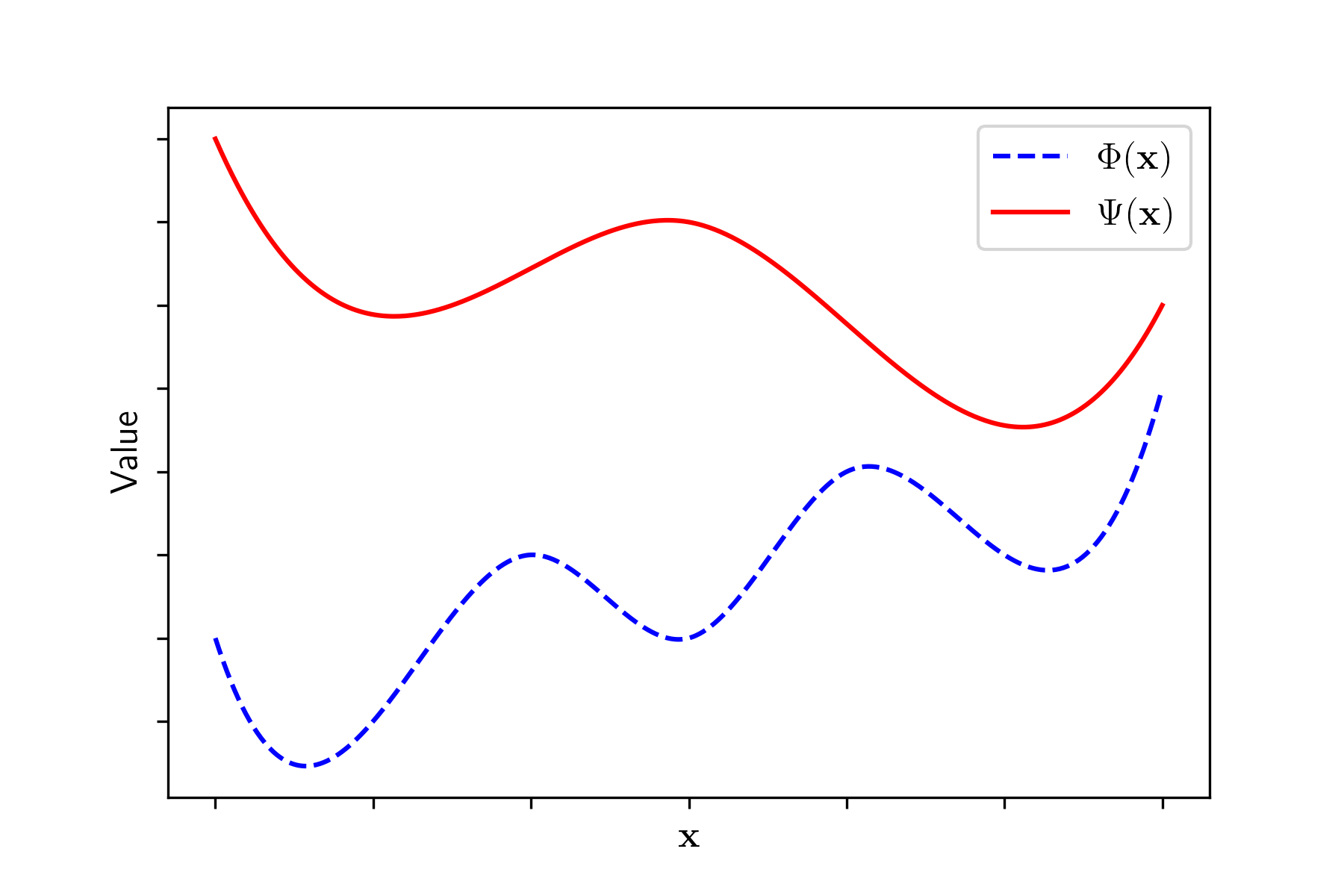}
    \caption{An illustration of $\Phi(\bm{x})$ and $\Psi(\bm{x})$.}
    \label{fig:potential}
\end{figure}

Next, we will explain how trust limit affects the convergence of diffusion using functions $\Phi(\bm{x})$ and $\Psi(\bm{x})$.
Fig. \ref{fig:potential} is an illustration of $\Phi(\bm{x})$ and $\Psi(\bm{x})$. The horizontal axis represents state $\bm{x}$, and the vertical axis represents function value. The global maximizer of $\Psi(\bm{x})$ is the all-$A$ state, and the global maximizer of $\Phi(\bm{x})$ is the risk-dominant equilibrium (which is all-$B$ state in this case). Any local maximizer of $\Phi(\bm{x})$ indicates a Nash equilibrium. Under the NE stochastic logit dynamics, the diffusion stochastically converges to the risk-dominant equilibrium that globally maximizes the potential function $\Phi(\bm{x})$. When players start to behave trustworthy, at each state, when the player's trust limit is so small that only the greedy best-response satisfies the constraint in \ref{eq:lt-best-resp-LP}, the diffusion progresses in the direction that maximizes $\Phi(\bm{x})$ and stochastically converges to the same direction as NE dynamics. When the player's trust limit is sufficiently large and both alternatives satisfy the constraint in \ref{eq:lt-best-resp-LP}, the LTE dynamics pushes the diffusion to progress towards the direction that maximizes $\Psi(\bm{x})$, which may result in a different convergence direction from the NE dynamics. Such a convergence direction is efficient because it maximizes social welfare globally.

\section{Mean Field Approximation for Diffusion Process and Markov Chain Analysis}\label{sec:mean-field-approx}

In this section, we study the following problem: given a certain number of current users in the social network, what is the probability that the superior alternative dominates the network or dies out eventually? And how rapidly will this take place? These questions are important from a managerial perspective as they allow companies to know how many advertising resources they should invest in the market such that their product prevails with high confidence. A reasonable decision is to allocate resources to influential individuals with high connectivity. However, companies often may not have a complete picture of the network's structure, especially if they expand into new markets. Even if the network structure is fully known, it is computationally difficult to determine the most influential individuals on a large-scale network, hence justifying the use of approximation schemes. 
Thus, we study the problem pessimistically by assuming that the network structure is unknown. We first introduce the mean-field approximation of diffusion on a class of graphs and then convert the diffusion process to the evolution of a reduced-size Markov chain.
By analyzing such a Markov chain, we study the outcome and convergence rate of the diffusion dynamics to answer the above questions.


\subsection{Mean Field Approximation}\label{subsec:mft}

The mean-field theory provides a powerful tool to study the diffusion dynamics over large-scale networks with random topologies \cite{benaim2003deterministic,lopez2006contagion}. We study the diffusion dynamics over a class of networks with the same degree sequence without explicitly specifying the network structure.
To describe our mean-field analysis, we partition the set of players into groups based on their degrees. All interactions are approximated by an average interaction and represented by the \emph{mean-field parameter}, same for all vertices independent of their connectivity and position. Moreover, it is assumed that there is no spatial correlation among the players choosing $A$ or $B$ across time. Hence, based on the assumptions, given an arbitrary player with degree $k$, each of its neighbors is in state $A$ with probability $\theta$, and the number of neighbors in state $A$ follows a binomial distribution $Bin(k,\theta)$.

Suppose that the network has $n$ vertices and degree distribution $(P(k))_{k=1}^K$, where $K$ is the maximum degree of the network. Without loss of generality, suppose that the graph is connected; otherwise, we can analyze each component separately. Let $\langle k \rangle = \sum_{k=1}^K k P(k)$ denote the average degree of graph $G$. Denote by $f(k)$ the proportion of players with degree $k$ in state $A$. Define the mean-field parameter $\theta$ as follows:
\begin{equation} \label{eq:mft-mf-par}
    \theta = \frac{1}{\langle k \rangle} \sum_{k=1}^K k P(k) f(k).
\end{equation}
$\theta$ is the probability that a given link is incident to a player in state $A$. It also approximates all interactions in the network.
Given a player with degree $k$ and trust limit $\delta_k$, define the \emph{normalized trust limit}
\begin{equation}
    \delta_k' = \frac{\delta_k}{k},
\end{equation}
which represents the average trust level that the player distributes to each of its neighbors. 
It is assumed that players with the same degree have the same normalized trust limit as they have similar social connectivity. Players of degree $k$ have normalized trust limit $\delta_k'$. However, the normalized trust limit can be different for different degree groups, and $\delta_k'$ can be different for different values of $k$.

Given the payoff matrix, a player with degree $k$, $k_A$ neighbors in state $A$ and normalized trust limit $\delta_k'$, the utility and social welfare can be computed, and hence the probability that the player chooses $A$ can be calculated using the NE diffusion model by Eq. \eqref{eq:nash-diffusion-logit-prob} or the LTE diffusion model by Eqs. \eqref{eq:lte-diffusion-1-prob}--\eqref{eq:lte-diffusion-2b-prob} as stated in Section \ref{sec:model}. 
Denote this probability by \(P(A|k,k_A,\delta_k')\), where we recall that $k_A \sim Bin(k,\theta)$. If we denote by \(Q(A|k,\theta,\delta_k')\) the overall probability that a player with degree $k$ and normalized trust limit $\delta_k'$ chooses $A$, we can compute \(Q(A|k,\theta,\delta_k')\) by taking the expectation of \(P(A|k,k_A,\delta_k')\) with respect to $k_A$ to get
\begin{align} 
    & Q(A|k,\theta,\delta_k') = \mb{E}_{k_A \sim Bin(k,\theta)}[P(A|k,k_A,\delta_k')] \notag \\
    & \quad = \sum_{k_A=0}^k P(A|k,k_A,\delta_k') \binom{k}{k_A} \theta^{k_A} (1-\theta)^{(k-k_A)}. \label{eq:mft-Q-defn}
\end{align}
As the probability term $P(A|k,k_A,\delta_k')$ embodies diffusion dynamics and the binomial term $\binom{k}{k_A} \theta^{k_A} (1-\theta)^{(k-k_A)}$ embodies the network structure, Eq. \eqref{eq:mft-Q-defn} connects the diffusion model and the social network in a convenient fashion.

\subsection{Markov Chain for Diffusion Process} \label{subsec:Markov-chain}

Based on the mean-field approximation, the diffusion process can be modeled as a discrete-time Markov chain. Given a degree distribution $(P(k))_{k=1}^K$ and the number of vertices $n$ of the graph, one can compute the degree sequence. Let \(\bm{n} = (n_1, n_2, \cdots, n_K)\) be a $K$-dimensional tuple such that $n_k$ is the number of vertices of degree $k$, i.e., $n_k = nP(k)$. Let \(\mc{S} = \{0, 1, \cdots, n_1\} \times \{0, 1, \cdots, n_2\} \times \cdots \times \{0, 1, \cdots, n_K\}\), and define \(\bm{y} = (y_1, y_2, \cdots, y_K)\), where $y_k \in \{0, 1, \cdots, n_k\}$ is the number of players with degree $k$ in state $A$. Clearly, $\mc{S}$ is the feasible set for $\bm{y}$. One can construct a Markov process $\bm{Y}(t) = (Y_1(t), \cdots, Y_K(t))$ on state space $\mc{S}$.\footnote{The Markov chain model can be generalized to the case of higher-order Markovity in which players have memory of the last $r$ states where $r \geq 1$. Define \emph{concatenated state} $\widetilde{\bm{Y}}(t) \triangleq (\bm{Y}(t-r+1), \bm{Y}(t-r+2), \cdots, \bm{Y}(t))$ for $t \geq r$. When $t < r$, the first $r-t$ entries of $\widetilde{\bm{Y}}(t)$ are filled with $\varnothing$. The random process $\widetilde{\bm{Y}}(t)$ on state space $\{\mc{S} \cup \{\varnothing\}\}^{r}$ is a Markov process.} The size of state space is \(|\mc{S}| = \prod_{k=1}^K (n_k+1)\), and thus \(\Omega(n) \leq |\mc{S}| \leq O(2^n)\) where $\Omega$ and $O$ are the notations for asymptotic lower and upper bounds respectively.

Given a state $\bm{y} = (y_1, \cdots, y_K)$, there are at most $(2K+1)$ transitions from $\bm{y}$. At most $K$ of them are \emph{forward transitions}: \((y_1, \cdots, y_k, \cdots, y_K) \rightarrow (y_1, \cdots, y_{k-1}, y_k+1, y_{k+1}, \cdots, y_K)\) for some $k\in [K]$; at most $K$ of them are \emph{backward transitions}: \((y_1, \cdots, y_k, \cdots, y_K) \rightarrow (y_1, \cdots, y_{k-1}, y_k-1, y_{k+1}, \cdots, y_K)\) for some $k\in [K]$; one is a self-loop \((y_1, \cdots, y_K) \rightarrow (y_1, \cdots, y_K)\). An illustration of forward and backward transitions and self-loop is shown in Fig. \ref{fig:mc-transition}.
A forward transition implies that one more player adopts the superior alternative $A$, whereas a backward transition implies that one fewer player adopts the superior alternative.
Using the mean-field approximation, one can derive the transition probabilities between different states of the Markov chain as: 
\begin{align}
    & P[(y_1, \cdots, y_K), (y_1, \cdots, y_{k-1}, y_k+1, y_{k+1}, \cdots, y_K)] \notag \\
    & \quad = \frac{n_k - y_k}{n} Q(A|k, \theta(\bm{y}), \delta_k'), \label{eq:mc-forw-trans-prob} \\
    & P[(y_1, \cdots, y_K), (y_1, \cdots, y_{k-1}, y_k-1, y_{k+1}, \cdots, y_K)] \notag \\
    & \quad = \frac{y_k}{n} \left(1 - Q(A|k, \theta(\bm{y}), \delta_k')\right), \label{eq:mc-back-trans-prob} \\
    & P[(y_1, \cdots, y_K), (y_1, \cdots, y_K)] \notag \\
    & \quad = 1 - \frac{n_k - 2y_k}{n} Q(A|k, \theta(\bm{y}), \delta_k') - \frac{y_k}{n}, \label{eq:mc-self-loop-prob} 
\end{align}
where the mean-field parameter $\theta$ is now a function of state $\bm{y}$ and can be calculated as 
\begin{equation} \label{eq:mft-mc-theta}
    \theta(\bm{y}) = \frac{\sum_{k=1}^K k y_k}{n \langle k \rangle}.
\end{equation}
Combining Eqs. \eqref{eq:mft-Q-defn}--\eqref{eq:mft-mc-theta}, one can characterize the one-step transition probability matrix of the Markov chain. 
The Markov chain can be viewed as a random walk with heterogeneous transition probabilities on a $K$-dimensional lattice graph with the length of the $k$-th dimension being $(n_k + 1)$. This lattice graph is named the \emph{underlying graph} of the Markov chain. A Markov chain is $k$-dimensional if its underlying graph is $k$-dimensional.

\begin{figure}[ht]
    \captionsetup{font=small}
    \centering
    \includegraphics[width=.4\textwidth]{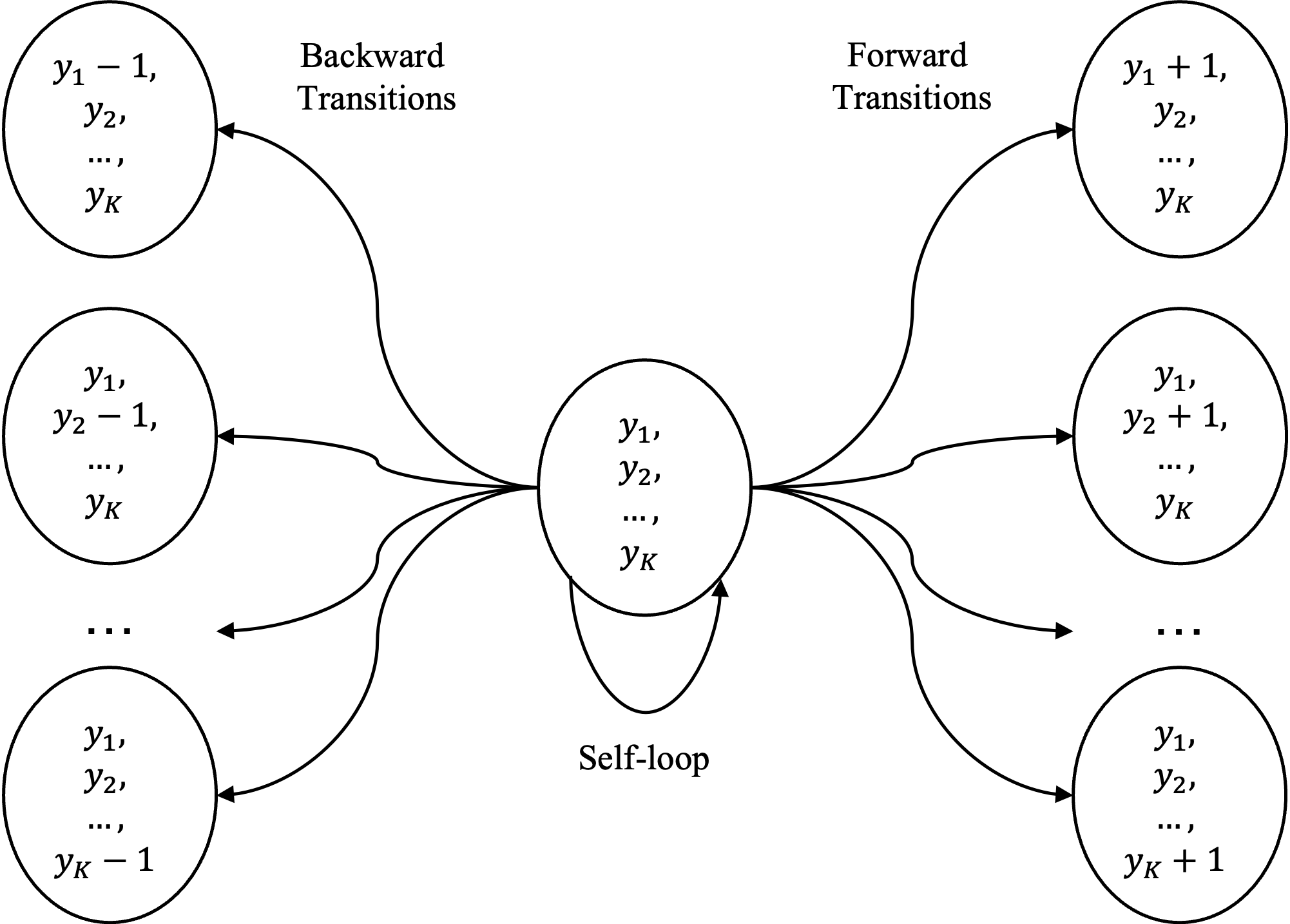}
    \caption{Illustration of forward and backward transitions and self-loop from state $(y_1, \cdots, y_K)$. (Transitions from other states are not shown.)}
    \label{fig:mc-transition}
\end{figure}

Two properties of the diffusion process are of interest: the \emph{probability of domination of the superior alternative} and the \emph{expected time to equilibrium}. The first one answers the question of how likely is the superior alternative to dominate the social network, and the second one answers the question of how fast, on average, the superior alternative dominates the market or dies out. We assume that the diffusion process terminates once either $A$ or $B$ completely dominates the network. In this regard, there are two absorbing states in the above Markov chain: the state $\bm{y}=\bm{0}$ and the state $\bm{y}=(n_1, \cdots, n_K)$. All other states are transient. Hence, given an arbitrary initial state, the Markov chain will eventually be absorbed into one of the absorbing states; that is, all individuals in the social network will adopt either $A$ or $B$. The \emph{absorption probability} and the \emph{expected time to absorption} are two important properties of absorbing Markov chains. The absorption probability to the superior-alternative-dominating state $\bm{y}=(n_1, \cdots, n_K)$ (i.e., all-$A$ state) reflects the probability of domination of the superior alternative; the expected time to absorption corresponds to the expected time to reach an equilibrium. 

The absorption probability and the expected time to absorption can be calculated by solving a system of linear equations. Let $\bm{P}$ denote the transition probability matrix of the Markov chain, and $\bm{T}$ and $\bm{R}$ be the submatrices of $\bm{P}$ which represent the transition among transients states and from transient states to absorbing states, respectively. The absorption probability $\bm{W}$ and expected time to absorption $\bm{\tau}$ can be calculated by solving \(\bm{W} = \bm{T}\bm{W} + \bm{R}\) and \(\bm{\tau} = \bm{T}\bm{\tau} + \bm{1}\). In Section \ref{sec:simulation}, we present the results of solving these two systems of equations and the results of simulation on a class of graphs to show that mean-field approximation estimates the convergence property of LTE diffusion dynamics reasonably well.

An important special case is when the underlying social network is a $k$-regular graph, in which case the above Markov chain becomes a random walk on the path and reduces to a 1-dimensional Markov chain. The state space is \(\mc{S} = \{0, 1, \cdots, n\}\), and each state \(y\in \mc{S}\) represents the number of agents adopting the superior alternative in the network. Given an arbitrary state $y\in \mc{S}$, the mean-field parameter is $\theta(y) = y/n$. Absorption probability and expected time to absorption for such Markov chain can be characterized in a close form as follows. Let $a_y = P[y,y+1]$ be the forward transition probability of the 1-dimensional Markov chain, and $b_y = P[y,y-1]$ be the backward transition probability. For simplicity, we define $a_0 = b_0 = a_n = b_n = 0$ since states $0$ and $n$ are absorbing states. The one-step transition probability matrix $\bm{P}\in \mb{R}^{(n+1)\times(n+1)}$ can be represented as follows:
\begin{equation} \label{eq:mc-trans-prob-mat} 
    P_{i,j} = 
    \begin{cases}
        a_i, & \text{if $j=i+1$}; \\
        b_i, & \text{if $j=i-1$}; \\
        1-a_i-b_i, & \text{if $i=j$}; \\
        0, & \text{otherwise}.
    \end{cases}
\end{equation}
Denote by $W^0_y$ the probability that the chain is absorbed in state 0 given the initial state $y$ and by $W^n_y$ the probability that the chain is absorbed in state $n$ given the initial state $y$. Also, let $\tau_{y}$ be the expected time that the chain is absorbed given the initial state $y$. The following theorem gives a characterization of $W^n_y$, $W^0_y$, and $\tau_y$.

\begin{theorem} \label{thm:absorbing}
Let 
\begin{equation*}
    \Delta = \sum_{i=1}^n \prod_{j=i}^{n-1} a_j \prod_{j=1}^{i-1} b_j.
\end{equation*}
Then $W^n_y$, $W^0_y$ and $\tau_y$ can be characterized as follows:
\begin{align}
    W_y^n & = \frac{1}{\Delta} \sum_{i=1}^y \prod_{j=i}^{n-1} a_j \prod_{j=1}^{i-1} b_j, \label{eq:thm-prob-n} \\
    W_y^0 & = 1 - P_y^n = \frac{1}{\Delta} \sum_{i=y+1}^n \prod_{j=i}^{n-1} a_j \prod_{j=1}^{i-1} b_j, \label{eq:thm-prob-0} \\
    \pi_y & = \frac{1}{\Delta} \Big[\sum_{j=1}^{n-1} \Big(\sum_{k=1}^{\min\{y,j\}}\prod_{l=k}^{j-1} a_l \prod_{l=1}^{k-1} b_l\Big) \notag \\
    & \qquad\qquad \Big(\sum_{k= \max\{y,j\}+1}^n \prod_{l=k}^{n-1} a_l \prod_{l=j+1}^{k-1} b_l\Big) \Big]. \label{eq:thm-time}
\end{align}
\end{theorem}

\begin{remark}
For non-regular graphs, one cannot expect to find a closed-form expression for the absorption probability and the expected absorption time as, in general, they involve an exponential number of terms. Nevertheless, using the same approach to solving a system of linear equations
\(
    \bm{W} = \bm{T}\bm{W} + \bm{R}
\)
and 
\(
    \bm{\tau} = \bm{T}\bm{\tau} + \bm{1}
\),
one can find these quantities for general chains in polynomial time in terms of the problem parameters.
\end{remark}



One application of the Markov chain analysis for LTE diffusion model is to address the seed placement problem. On the one hand, given an allocation of ``free samples'' to the players based on their degree, which is equivalent to the initial state of the Markov chain, one can compute the absorption probability starting from that initial state to estimate the likelihood of domination of the superior alternative. On the other hand, given a desired confidence level $\alpha$, one can compute how many initial seeds (adopters) are needed and how the free samples are allocated among the players by looking up the initial states in $\bm{W}$. In a 1-dimensional Markov chain, the initial states from which the Markov chain is absorbed in state $n$ (all-$A$ state) with a probability of at least $\alpha$ are given by \(    \left\{ y\in \{0,1,\ldots,n\}: W_y^n \geq \alpha\right\}\),
where $W_y^n$ is the probability that the chain is absorbed in state $n$ given initial state $y$, the expression of which is given by Eq. \eqref{eq:thm-prob-n}. To find the minimum number of free samples, one can find the minimum value of $y$ that belongs to the set. In a high-dimensional Markov chain, not only the number of free samples, but also the distribution among the players with different degrees affects the absorption probability. Denote by $W^{\bm{n}}_{\bm{y}}$ the absorption probability to the state $(n_1, \cdots, n_K)$ starting from an arbitrary state $\bm{y} = (y_1, \cdots, y_K) \in \mc{S}$. Recall that $\mc{S}$ is the state space and \(\mc{S} = \{0, 1, \cdots, n_1\} \times \{0, 1, \cdots, n_2\} \times \cdots \times \{0, 1, \cdots, n_K\}\).
In that case, the initial states from which the Markov chain is absorbed in all-$A$ state $(n_1, \cdots, n_K)$ with probability at least $\alpha$ are given by \(\left\{ \bm{y} \in \mc{S}: W^{\bm{n}}_{\bm{y}} \geq \alpha \right\}\).

\section{Threshold Diffusion Model under Limited-trust Equilibrium}\label{sec:threshold-model}

As we saw in Section \ref{sec:preliminaries}, the LT model is a well-studied diffusion model for influence maximization over social networks. In this section, we will show that the LTE diffusion model can be converted to a non-progressive LT model if all players update their states deterministically and synchronously by following the best-response rule. Another advantage of such a reduction to the LT model is that it allows us to study the convergence behavior of LTE dynamics under certain assumptions using the rich set of results on the LT model. 


Let us consider the case where all players update their purchase option synchronously at each time instant and strictly follow the best-response strategy on the notion of LTE. That being said, there is no randomness in the diffusion model. Using the same notation as in Section \ref{sec:model}, the \emph{best-response LTE (BR-LTE) diffusion model} can be described as follows: 

\medskip
\noindent
\underline{BR-LTE Diffusion Model}
\begin{enumerate}[label=\alph*)]
    \item \(u'_i(\bm{x}_{-i}) \geq u^*_i(\bm{x}_{-i}) - \delta_i\), 
    \begin{equation} \label{eq:br-lte-diffusion-a-prob}
        P(x_i \in W_i(\bm{x}_{-i})|\bm{x}_{-i}) = 1;
    \end{equation}
    
    \item \(u'_i(\bm{x}_{-i}) < u^*_i(\bm{x}_{-i}) - \delta_i\),
    \begin{equation} \label{eq:br-lte-diffusion-b-prob}
        P(x_i \in U_i(\bm{x}_{-i})|\bm{x}_{-i}) = 1.
    \end{equation}
\end{enumerate}
where $u'_i(\bm{x}_{-i})$, $u^*_i(\bm{x}_{-i})$, $W_i(\bm{x}_{-i})$ and $U_i(\bm{x}_{-i})$ are defined as in Section \ref{sec:model}. In fact, the probability functions of BR-LTE diffusion model are equivalent to the LTE diffusion model by setting $\beta \to \infty$ and $\beta' \to \infty$ such that the randomness represented by $1/\beta$ and $1/\beta'$ approaches zero.

The BR-LTE diffusion model can be converted to a non-progressive LT (NP-LT) model. In an NP-LT model, an active player is allowed to become inactive again if the proportion of the active neighbors is less than its threshold. This conversion establishes a connection between the epidemic and the game-theoretic diffusion models. Let product $A$ be the active state and product $B$ be the inactive state. We use the notation of $k^i$ and $k_A^i$ for the number of neighbors and the number of neighbors in state $A$ of player $i$, and $k_A^i/k^i$ to represent the proportion of the player's neighbors in state $A$. Define $q^*_i$ as the threshold that player $i$ becomes active. If $k_A^i/k^i \geq q^*_i$, player $i$ will adopt $A$ (becomes active); if $k_A^i/k^i < q^*_i$, player $i$ will adopt $B$ (becomes inactive). $q^*_i$ depends on the utility of the products, which is embodied in the payoff matrix, and the player's normalized trust limit $\delta_i'$. In general, the higher utility the superior alternative brings about to the player and its neighbors, the smaller the threshold will be; the more trustworthy the player is, the lower threshold it will have. Therefore, the BR-LTE model is equivalent to the following LT model named \emph{linear threshold LTE (LT-LTE) diffusion model}:

\noindent
\underline{LT-LTE Diffusion Model}
\begin{equation} \label{eq:lte-threshold-prob}
    P(x_i = A | \bm{x}_{-i}) = 
    \begin{cases}
        0 \quad \text{if $k_A^i/k^i < q^*_i$} \\
        1 \quad \text{if $k_A^i/k^i \geq q^*_i$}
    \end{cases}
\end{equation}
where 
\begin{align}\label{eq:lte-threshold-q-star}
    q^*_i = \min \Bigg\{ & \max \left\{0, \frac{2b-c-d}{2(a+b-c-d)}, \frac{b-c-\delta'_i}{a+b-c-d} \right\}, \notag \\ 
    & \frac{b-c+\delta'_i}{a+b-c-d} \Bigg\}.
\end{align}
Now, let us define 
\begin{align}\nonumber
    q_w &= \max \left\{0, \frac{2b-c-d}{2(a+b-c-d)}\right\}, \quad q_u = \frac{b-c}{a+b-c-d},\cr
    \widetilde{\delta}'_i &= \frac{\delta'_i}{a+b-c-d}.
\end{align}
Here, $q_w$ represents the proportion of neighbors in state $A$ such that $A$ and $B$ yield same social welfare. Similarly, $q_u$ represents the proportion of neighbors in state $A$ such that $A$ and $B$ yield same utility to player. Moreover, $\widetilde{\delta}$ represents the relative normalized trust limit of player $i$ with respect to the payoff matrix. Then, the threshold $q^*_i$ can be simplified as 
\begin{equation}\label{eq:lte-threshold-q-star-sim}
    q^*_i = \min \left\{ \max\left\{q_w, q_u-\widetilde{\delta}'_i\right\}, q_u+\widetilde{\delta}'_i\right\}. \tag{\ref{eq:lte-threshold-q-star}a}
\end{equation}

It is easy to see that $q^*_i$ is monotonically nonincreasing with respect to $\delta'_i$ or $\widetilde{\delta}'_i$, implying that players with high normalized trust limit have thresholds not greater than the players with low normalized trust limit. For all payoff matrices, the range of $q_w$ is $[0,1/2)$, and the range of $q_u$ is $(0,1)$. If $q_w < q_u$, $q^*_i\in [q_w, q_u)$ for all $\delta'_i>0$; if $q_w > q_u$, $q^*_i\in (q_u, q_w]$ for all $\delta'_i>0$; if $q_w = q_u$, $q^*_i$ is a constant and $q^*_i = q_w$ independent of the player's trust limit. Note that in the BR-LTE diffusion model, one needs the information about the payoff matrix, the player's trust limit, and the number of its neighbors to compute the best-response strategy numerically. In the LT-LTE model, this information is encoded in a single scalar parameter $q^*_i$.

\section{Simulation Results}\label{sec:simulation}

In this section, we provide various simulation results on the LTE diffusion model. First, we compare the diffusion evolution of NE and LTE dynamics on ego-Facebook social network \cite{snapnets}. Next, we compare the results obtained from Markov chain analysis in Section \ref{sec:mean-field-approx} and the numerical results of LTE diffusion dynamics on a class of graphs. 
Additional simulation results are available in the e-companion \cite{leon2022limitedtrust-arxiv}.


\subsection{Comparison between NE and LTE diffusion} \label{subsec:diff-sim}

In this subsection, we compare the diffusion process under NE and LTE dynamics. We simulate the diffusion evolution of the superior alternative with payoff matrices in Table \ref{tab:payoff-sim-1} and \ref{tab:payoff-sim-2}. The sensitivity parameters are $\beta = \beta' = 0.5$. The simulation is carried out on ego-Facebook social network \cite{snapnets} that is constructed from real-dataset to represent interactions between users of Facebook. The social network contains 4,039 nodes and 88,234 undirected edges. All players in the social network have the same normalized trust limit $\delta'$.
We consider various values of $\delta'$ from 0.5 to 2 and various numbers of initial adopters of superior alternative A (5\%, 10\%, 20\%, and 50\% of the total number of players). For each setting, the average result of 100 trials is taken. Figure \ref{fig:sim-egofb-adop} shows the average proportion of adoption of the superior alternative as diffusion progresses; Fig.\ \ref{fig:sim-egofb-util} shows the players' average utility versus diffusion time. Both figures show consistent trend of diffusion evolution.

\begin{table}[h]
    \captionsetup{font=small}
    \caption{Payoff matrices for simulation.}
    \begin{subtable}[h]{0.225\textwidth}
    \caption{}
    \centering
    \begin{tabular}{|c||c|c|}
        \hline
          & $A$ & $B$  \\ \hline \hline
        $A$ & 8, 8 & 2, 7 \\ \hline
        $B$ & 7, 2 & 4, 4 \\ \hline
    \end{tabular}
    \label{tab:payoff-sim-1}
    \end{subtable}
    \hfill
    \begin{subtable}[h]{0.225\textwidth}
    \caption{}
    \centering
    \begin{tabular}{|c||c|c|}
        \hline
          & $A$ & $B$  \\ \hline \hline
        $A$ & 8, 8 & 2, 5 \\ \hline
        $B$ & 5, 2 & 4, 4 \\ \hline
    \end{tabular}
    \label{tab:payoff-sim-2}
    \end{subtable}
    \label{tab:payoff-sim-exmp}
\end{table}

As shown in Figs.\ \ref{subfig:sim-egofb-adop-no-sw-gap} and\ \ref{subfig:sim-egofb-util-no-sw-gap}, under the NE dynamics, the superior alternative dies out regardless of the number of initial adopters, because $B$ is the risk-dominant strategy ($b-c = 2 > 1 = a-d$). This may happen when the new product $A$ has poor compatibility with product $B$. As the players start to behave trustworthy, though $B$ still prevails when $\delta'$ is small ($\delta' = 0.5$ and $1$ when the number of initial adopters is less than 20\% of the total number of players), the convergence rate is slower than the NE dynamics. As the players increase their trust limit $\delta'$ to 1.5, the superior alternative starts to dominate, and the diffusion converges to the all-$A$ state. When $\delta'$ becomes large ($\delta' = 2$), the superior alternative rapidly spreads over the network. 
In Figs.\ \ref{subfig:sim-egofb-adop-sw-gap} and\ \ref{subfig:sim-egofb-util-sw-gap}, under NE diffusion dynamics, the diffusion evolution does not converge to the all-A state when the number of initial adopters is only 5\%, 10\% and 20\% of the total number of players even though A is the risk-dominant strategy according to the payoff matrix in Table \ref{tab:payoff-sim-2}. However, the LTE diffusion dynamics are robust and drive the diffusion process towards the all-A state when the trust limit is moderately large. Similar to the previous case, under the LTE diffusion dynamics, when the number of initial adopters is 5\%, 10\%, and 20\% of the total number of players and when the trust limit is small ($\delta' = 0.5$), the LTE diffusion converges to the same equilibrium state where A dies out as the NE diffusion dynamics. As the trust limit increases, the equilibrium state completely changes to the state where almost all players adopt A. When the number of initial adopters is 50\% of the total number of players in the social network, under both NE and LTE diffusion dynamics, the diffusion process converges to the all-A state. Under the LTE diffusion dynamics, the superior alternative spreads faster than the case of NE dynamics.
For both payoff matrices and in most cases, no matter whether the superior alternative finally dominates the social network or dies out, the LTE model ensures higher proportion of adoption of the superior alternative and higher average utility than the NE model, and the difference becomes larger as the players' trust limit increases in a certain range. The only exception is the case when the payoff matrix is the one in Table \ref{tab:payoff-sim-2}, $\delta'=0.5$ and the number of initial adopters is 5\% of the total number of players shown in Fig. \ref{subfig:sim-egofb-adop-sw-gap}, and the LTE diffusion dynamics converges to a state with slightly lower adoption of the superior alternative than the NE dynamics. However, by comparing the results in Fig. \ref{subfig:sim-egofb-util-sw-gap}, we can see that the average utility of the LTE dynamics is as large as the NE dynamics. The slightly lower adoption portion may result from the redistribution of the superior alternative among high- and low-degree players, which does not adversely affect the players' average utility.

The results from Fig.\ \ref{fig:sim-egofb-util} show that when players behave trustworthy, the average utility is always greater than the case when players are solely self-interested as diffusion progresses. This difference becomes even more prominent when the trust limit increases. 
In Fig.\ \ref{subfig:sim-egofb-util-no-sw-gap} and for the payoff matrix in Table \ref{tab:payoff-sim-1}, during the diffusion, the average utility increases when players show moderate to large trust ($\delta' = 1.5$ and $2$), whereas it decreases when players are completely selfish (NE).
In Fig.\ \ref{subfig:sim-egofb-util-sw-gap} and for the payoff matrix in Table \ref{tab:payoff-sim-2}, similar results can be observed, and the average utility increases during the diffusion process when players show moderate to large trust. When the players are self-interested (NE), the average utility does not change significantly or increases slower than when the players are trustworthy.
Compared to Fig. \ref{fig:sim-egofb-adop}, these results suggest that the prevalence of the superior alternative is associated with the players' hope of obtaining higher long-term utility as a result of their neighbors' return in a dynamic setting.

\begin{figure}[h]
    \captionsetup{font=small}
    \centering
    \begin{subfigure}[h]{0.47\textwidth}
    \centering
    \includegraphics[trim=0.9in 0.25in 0.9in 0.25in,clip,width=\textwidth]{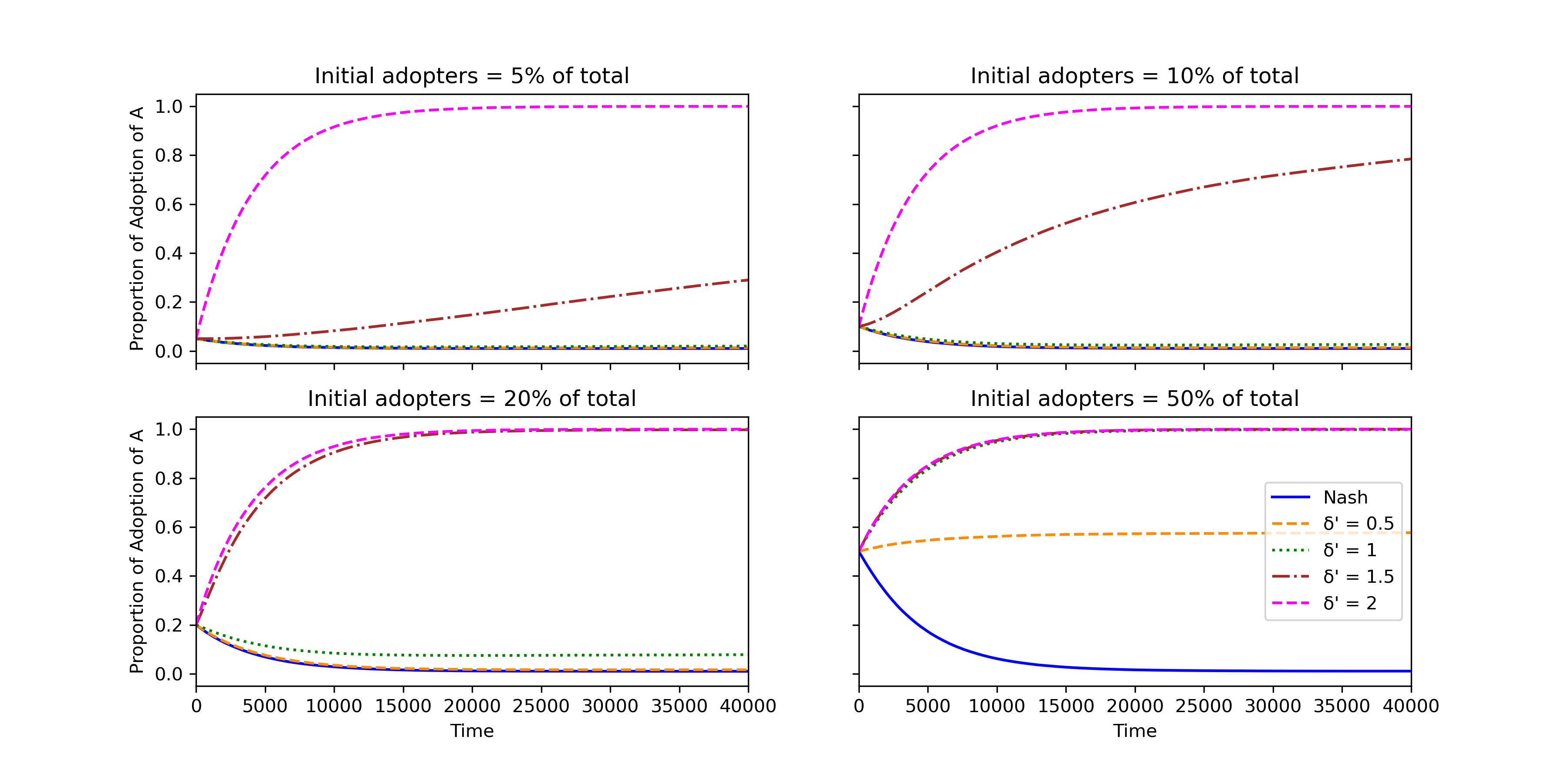}
    \subcaption{Payoff matrix in Table \ref{tab:payoff-sim-1}.}
    \label{subfig:sim-egofb-adop-no-sw-gap}
    \end{subfigure}
    \begin{subfigure}[h]{0.47\textwidth}
    \centering
    \includegraphics[trim=0.9in 0.25in 0.9in 0.25in,clip,width=\textwidth]{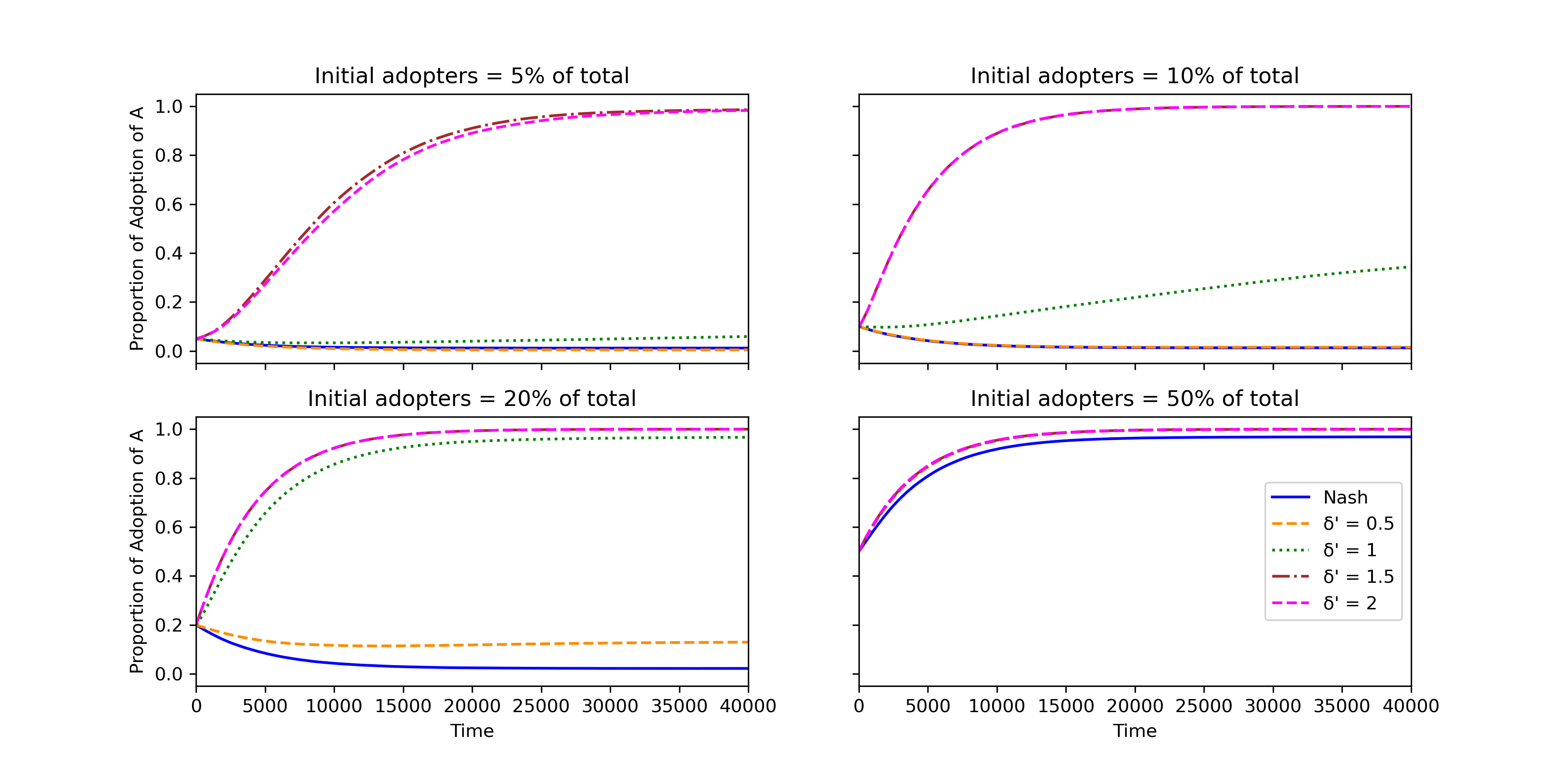}
    \subcaption{Payoff matrix in Table \ref{tab:payoff-sim-2}.}
    \label{subfig:sim-egofb-adop-sw-gap}
    \end{subfigure}
    \caption{Simulation results of diffusion on ego-Facebook social network: average proportion of adoption of A vs. time.}
    \label{fig:sim-egofb-adop}
\end{figure}

\begin{figure}[h]
    \captionsetup{font=small}
    \centering
    \begin{subfigure}[h]{0.47\textwidth}
    \centering
    \includegraphics[trim=0.9in 0.25in 0.9in 0.25in,clip,width=\textwidth]{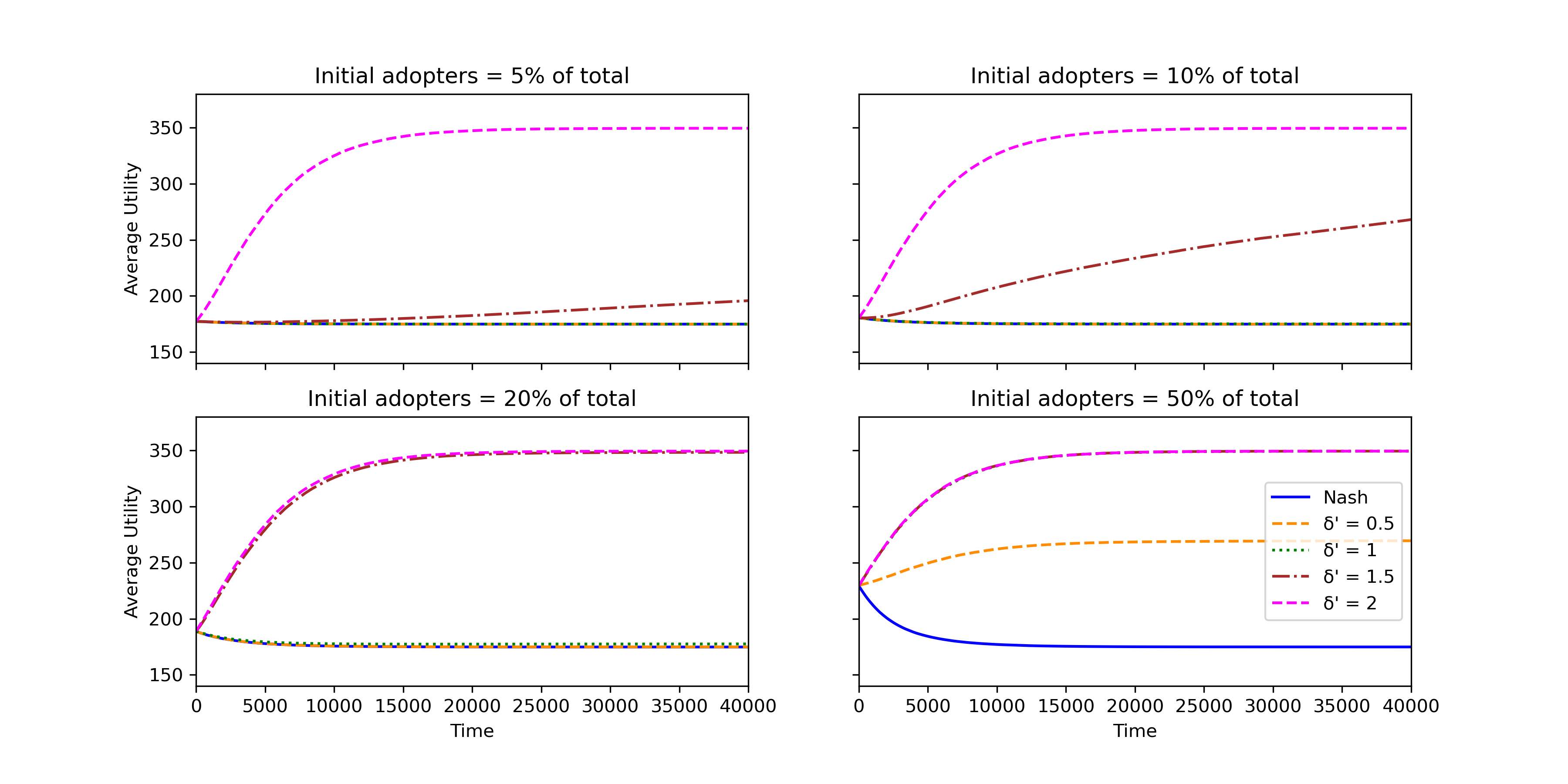}
    \subcaption{Payoff matrix in Table \ref{tab:payoff-sim-1}.}
    \label{subfig:sim-egofb-util-no-sw-gap}
    \end{subfigure}
    \begin{subfigure}[h]{0.47\textwidth}
    \centering
    \includegraphics[trim=0.9in 0.25in 0.9in 0.25in,clip,width=\textwidth]{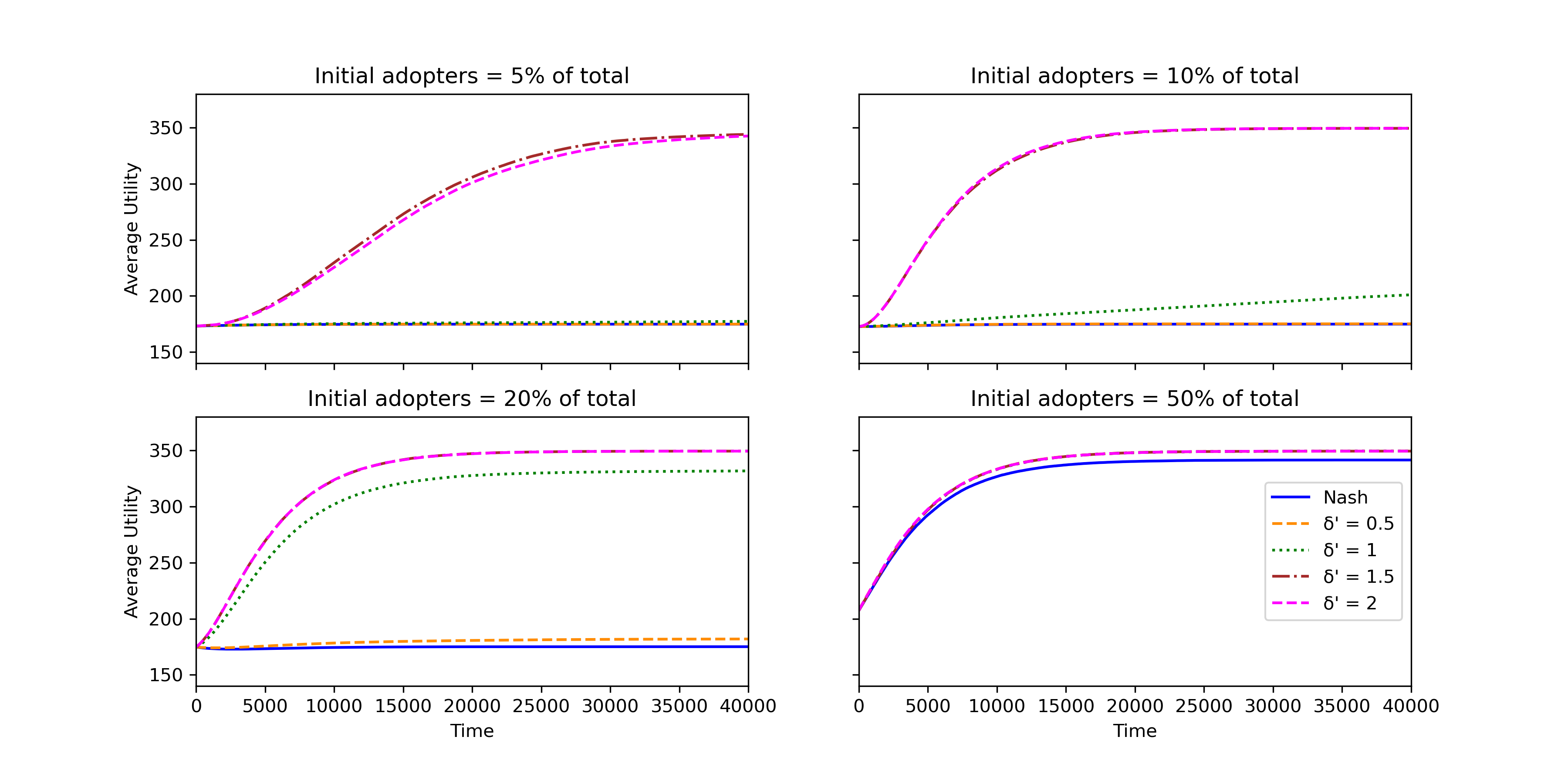}
    \subcaption{Payoff matrix in Table \ref{tab:payoff-sim-2}.}
    \label{subfig:sim-egofb-util-sw-gap}
    \end{subfigure}
    \caption{Simulation results of diffusion on ego-Facebook social network: average utility vs. time.}
    \label{fig:sim-egofb-util}
\end{figure}

\subsection{Comparison between LTE Diffusion Results and Markov Chain Analysis Results}

Next, we present the simulation results of diffusion over a class of networks specified by their degree distribution and compare them with the results obtained from the Markov chain analysis. We consider two types of degree distributions: (a) regular graphs, which correspond to a 1-dimensional Markov chain, and (b) graphs with two types of degree, which correspond to a 2-dimensional Markov chain.
Firstly, we simulate the diffusion on a class of 10-regular graphs with 1000 players.
All players have the same normalized trust limit $\delta'$.
The payoff matrix in Table \ref{tab:payoff-sim-1} is used, and various numbers of initial adopters (5\%, 10\%, 20\%, and 50\% of the total number of players) and values of normalized trust limit ($\delta'$ from 1 to 2) are used. The sensitivity parameters are $\beta=\beta'=1$. For each set of parameters, 100 trials are carried out, and a new network is randomly generated in each trial. The average result of 100 trials is taken. The results are shown in Fig.\ \ref{fig:sim-random-10-reg-comp}.

In Fig.\ \ref{fig:sim-random-10-reg-comp}, each color represents a set of three curves under the same $\delta'$. The curve in the solid line represents the diffusion simulation result; the horizontal dashed line and the vertical dotted line represent the absorption probability to the superior-alternative-dominating state (all-$A$ state) and the expected time to absorption computed from the Markov chain analysis respectively. Figure \ref{fig:sim-random-10-reg-comp} shows that the results from the Markov chain analysis accord with the simulation results. The diffusion (solid lines) converges to the equilibrium state indicated by the absorption probability (horizontal dashed lines), and equilibrium is roughly achieved at the expected time to absorption (vertical dotted lines) as the diffusion curves become flat. This shows that the 1-dimensional Markov chain model gives a good estimation of the convergence property of the LTE dynamics on random regular graphs.

\begin{figure}[h]
    \captionsetup{font=small}
    \centering
    \includegraphics[trim=0.9in 0.4in 0.9in 0.5in,clip,width=0.47\textwidth]{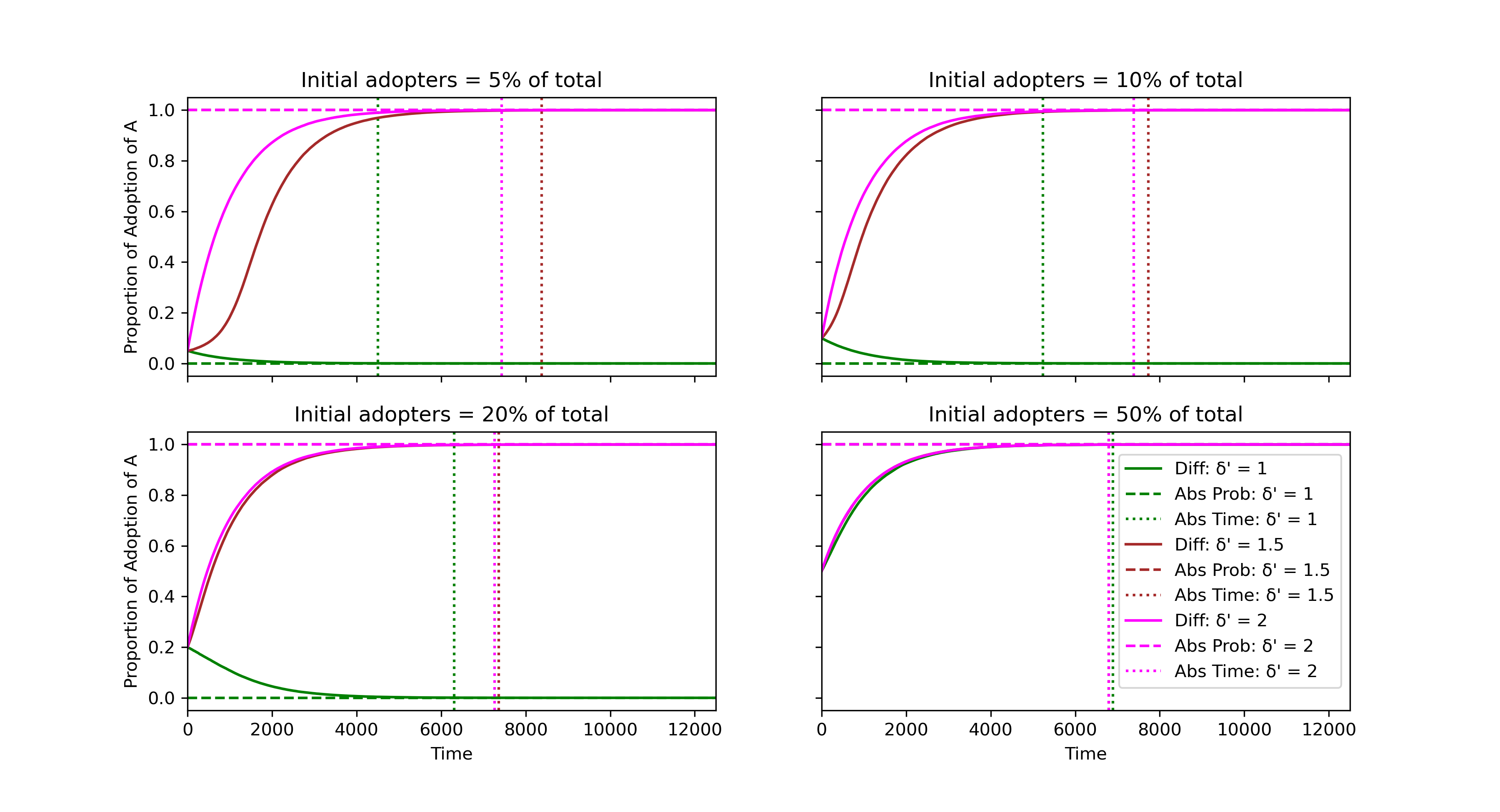}
    \caption{Simulation results of diffusion on a class of 10-regular graphs and computation results from 1d Markov chain analysis with payoff matrix in Table \ref{tab:payoff-sim-1}.}
    \label{fig:sim-random-10-reg-comp}
\end{figure}

Next, we simulate the diffusion over a class of networks with two types of degrees: 3 and 7. The corresponding Markov chain of such networks is 2-dimensional. There are 400 players in the network: 300 players with degree 3 and 100 players with degree 7. The average degree is 4. The payoff matrix in Table \ref{tab:payoff-sim-1} is used. The number of initial adopters are set to be 5\%, 10\%, 20\% and 50\% of the total number of players. The ``initial free samples'' are either proportionally distributed to high-degree and low-degree players or distributed to high-degree players as many as possible. All players have the same normalized trust limit $\delta'$.
Various values of normalized trust limit ($\delta'$ from 0.5 to 2) are used. The sensitivity parameters are $\beta=\beta'=1$. For each set of parameters, 100 trials are carried out, and a new network is randomly generated in each trial. The average result of 100 trials is taken. The results are shown in Fig.\ \ref{fig:sim-random-2d-comp}.

In Fig.\ \ref{fig:sim-random-2d-comp}, the diffusion simulation result, absorption probability to superior-alternative-dominating state, and expected time to absorption are plotted in solid, dashed, and dotted lines, respectively. Each color represents a value of $\delta'$. We can see from the results that when the absorption probability to the superior-alternative-dominating state (all-$A$ state) is close to 0 or 1 from Markov chain analysis, the Markov chain analysis matches the simulation result very well. When the absorption probability to all-A state is neither close to 1 nor close to 0, the prediction of the probability of domination of the superior alternative is less accurate. For example, in the bottom left plot of Fig. \ref{fig:sim-random-2d-comp}, when $\delta'=0.5$, the Markov chain analysis gives an intermediate value of absorption probability to the superior-alternative-dominating state. Nevertheless, the general trend of diffusion evolution of this case still matches the result of Markov chain analysis. The Markov chain model is still beneficial considering the simple construction and accurate estimation of diffusion process when the estimated probability of domination or dying out of the superior alternative is high.

\begin{figure}[h]
    \captionsetup{font=small}
    \centering
    \includegraphics[trim=0.9in 1.25in 0.9in 1.4in,clip,width=0.47\textwidth]{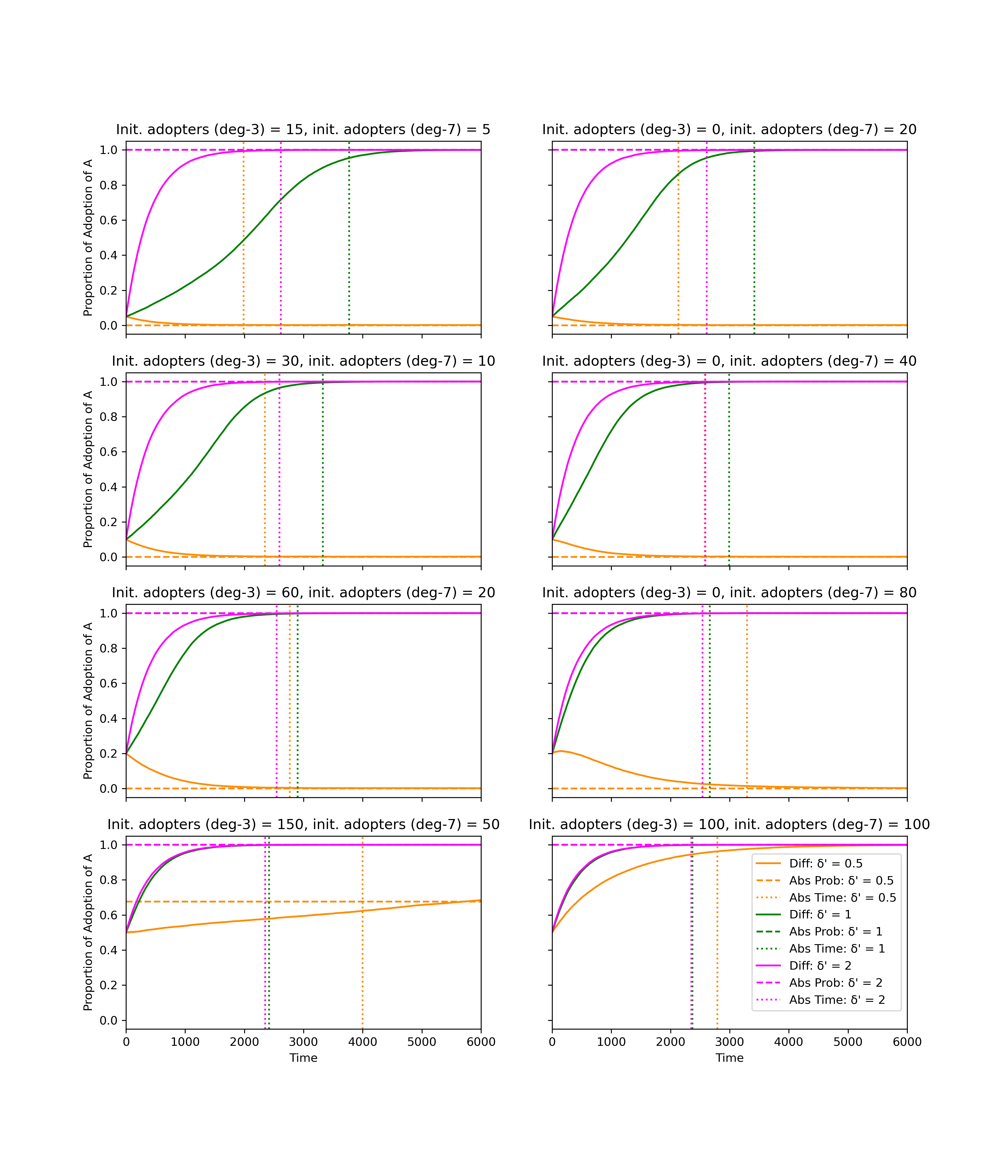}
    \caption{Simulation results of diffusion on random graph with two types of degree (2 and 4) and results from 2d Markov chain analysis with payoff matrix in Table \ref{tab:payoff-sim-1}.}
    \label{fig:sim-random-2d-comp}
\end{figure}

\section{Conclusion}\label{sec:conclusion}

In this paper, we analyzed diffusion of a superior alternative in social networks. Each individual has a trust limit and behaves trustworthy towards its neighbors for long-term benefit. In each iteration, an individual plays a coordination game with its neighbors based on the notion of limited-trust equilibrium with the hope of getting the return from its neighbors in future. 


To study whether the superior alternative prevails or dies out, the diffusion model is transformed into a reduced-size Markov chain based on mean-field approximation. The Markov chain represents a random walk on a lattice graph with heterogeneous transition probabilities and two absorbing states. The probability of domination of the  superior alternative and the expected time to equilibrium are represented by the absorption probability to the superior-alternative-dominating state and the expected time to absorption of the absorbing Markov chain and can be computed by solving a system of linear equations. In particular, for regular graphs, a closed-form expression can be obtained. The best-response LTE diffusion model can be transformed into a linear threshold diffusion model when all players update their state synchronously, and therefore, a connection between game-theoretic and epidemic diffusion models is established. Simulation results show that when players behave trustworthy, their long-term utility is significantly greater than the case when they are solely self-interested. The incentive for long-term utility also drives the superior alternative to spread over the network when their trust limit is large. Further simulation results show that the Markov chain analysis estimates the convergence property of LTE dynamics reasonably well.

The Markov chain analysis is useful from a managerial perspective. On the one hand, it is helpful for the companies to estimate the probability that the product or service dominates the social network or dies out and to estimate the time frame that such dominance or dying-out takes place. It is easy to use the Markov chain analysis to check whether the seed selection strategy can enable the alternative to spread over the social network with high confidence when the company has limited information about the social network structure. On the other hand, given the desired confidence level that the alternative dominates the social network, the Markov chain analysis can be used to design seed selection strategies. The simple form of the Markov chain analysis makes it useful for the companies to carry out a first-step estimation and strategy design.

This work opens several interesting future directions. For instance, how will diffusion progress when only a portion of agents behave trustworthy? And how can one obtain a high-probability bound for the convergence time of LTE diffusion to the equilibrium points?



\ifCLASSOPTIONcaptionsoff
  \newpage
\fi



\bibliographystyle{IEEEtran}
\bibliography{references}

\newpage

\begin{strip}
\centering
\Huge{
Limited-Trust in Diffusion of Competing Alternatives over Social Networks \\
E-Companion
}
\\
\vspace{.5em}
\large{
Vincent Leon, S. Rasoul Etesami, and Rakesh Nagi
}
\end{strip}

\appendices

\section{Numerical Examples} \label{ap:num-examples}


\subsection{Numerical Examples for LTE Diffusion Model}

Consider an arbitrary player $i\in[n]$. Let $k$ be the number of neighbors of $i$, and $k_A$ be the number of neighbors of $i$ in state $A$. Using the payoff matrix in Table \ref{tab:coord-game-payoff-matrix}, we can rewrite Eqs. \eqref{eq:nash-diffusion-logit-prob}, \eqref{eq:lte-diffusion-2a-prob} and \eqref{eq:lte-diffusion-2b-prob} as follows: 

\noindent
\underline{From Eqs. \eqref{eq:nash-diffusion-logit-prob} and \eqref{eq:lte-diffusion-2b-prob}:}
\begin{equation} \label{eq:nash-diffusion-logit-prob-payoff}
    P(x_i = A | \bm{x_{-i}}) = \frac{e^{\beta [k_A(a-d)+(k-k_A)(c-b)]}}{e^{\beta [k_A(a-d)+(k-k_A)(c-b)]}+1}.
\end{equation}

\noindent
\underline{From Eq. \eqref{eq:lte-diffusion-2a-prob}:}
\begin{equation} \label{eq:lte-diffusion-2a-prob-payoff}
    P(x_i = A | \bm{x}_{-i}) = \frac{e^{\beta' [k_A(2a-c-d)+(k-k_A)(c+d-2b)]}}{e^{\beta' [k_A(2a-c-d)+(k-k_A)(c+d-2b)]}+1}.
\end{equation}
As shown in \eqref{eq:nash-diffusion-logit-prob-payoff} and \eqref{eq:lte-diffusion-2a-prob-payoff}, the structure of payoff matrix plays a significant role in the diffusion process. In this appendix, we consider a few special cases of the payoff matrix. We assume that the player has the same sensitivity to utility and social welfare, i.e., $\beta = \beta'$.

{\bf $\diamond$ Case 1:}
Consider a player with $k$ neighbors in a network. Suppose that initially the player and its neighbors all have product $B$. If the player deviates from its neighbors and purchases $A$, it will suffer a loss of $k(b-c)$.
Under the NE diffusion model, the probability that the player changes its state from $B$ to $A$ is small. However, if by doing so its neighbors' utility increases, under the LTE diffusion model, the player will be incentivized to purchase $A$ provided that the decrease in utility is within its trust level. If the decrease in utility is beyond its trust limit, the player will play according to NE, and the state-updating probability is the same as the NE diffusion model. Therefore, in this case, the probability that the player changes its state to $A$ under the LTE diffusion model is at least as large as that under the NE diffusion model. Such payoff matrix has a structure that $c+d>2b$, resulting in a property that for any player in the network, choosing $A$ always yields higher social welfare regardless of the number and the states of its neighbors. This is because, along any edge connecting two interacting players, either player adopting $A$ yields a higher sum of utilities regardless of the other player's state. The payoff matrix with $c+d>2b$ is called payoff matrix \emph{without social welfare gap} as any player's deviation from $B$ to $A$ results in an increase in social welfare regardless of the others' states. An example of the payoff matrix without social welfare gap is shown in Table \ref{tab:payoff-no-sw-gap}.

\begin{table}[h]
    \captionsetup{font=small}
    \caption{Examples of payoff matrices without and with social welfare gap.}
    \begin{subtable}[h]{0.225\textwidth}
    \caption{Payoff matrix without social welfare gap.}
    \centering
    \begin{tabular}{|c||c|c|}
        \hline
          & $A$ & $B$  \\ \hline \hline
        $A$ & 8, 8 & 2, 7 \\ \hline
        $B$ & 7, 2 & 4, 4 \\ \hline
    \end{tabular}
    \label{tab:payoff-no-sw-gap}
    \end{subtable}
    \hfill
    \begin{subtable}[h]{0.225\textwidth}
    \caption{Payoff matrix with social welfare gap.}
    \centering
    \begin{tabular}{|c||c|c|}
        \hline
          & $A$ & $B$  \\ \hline \hline
        $A$ & 8, 8 & 2, 5 \\ \hline
        $B$ & 5, 2 & 4, 4 \\ \hline
    \end{tabular}
    \label{tab:payoff-sw-gap}
    \end{subtable}
    \label{tab:payoff-examples}
\end{table}

When there is no social welfare gap in the payoff matrix, given any state of the system, the probability that a player chooses product $A$ under the LTE diffusion model is at least as large as that under the NE diffusion model. A brief analysis is provided in the following, where subscripts NE and LTE are used to distinguish the state-updating probability under the NE and LTE diffusion models, respectively. 
\begin{enumerate}
    \item If product $A$ maximizes both social welfare and utility:
    \begin{equation*}
        1 = P_{\text{LTE}}(x_i = A | \bm{x}_{-i}) > P_{\text{NE}}(x_i = A | \bm{x}_{-i})
    \end{equation*}
    
    \item If $A$ maximizes social welfare and $B$ maximizes utility:
    \begin{enumerate}
        \item \(u_i(A, \bm{x}_{-i}) \geq u_i(B, \bm{x}_{-i}) - \delta_i\): 
        \begin{align*}
             \!\!\!\!\!\!P_{\text{LTE}}(x_i = A | \bm{x}_{-i})& = \frac{e^{\beta' [k_A(2a-c-d)+(k-k_A)(c+d-2b)]}}{e^{\beta [k_A(2a-c-d)+(k-k_A)(c+d-2b)]}\!+\!1} \\
            & > \frac{e^{\beta [k_A(a-d)+(k-k_A)(c-b)]}}{e^{\beta [k_A(a-d)+(k-k_A)(c-b)]}+1} \\
            & = P_{\text{NE}}(x_i = A | \bm{x}_{-i}),
        \end{align*}
        where the inequality holds because $\beta=\beta'$, $a>d$, $b>c$ and $c+d>2b$.
        
        \item \(u_i(A, \bm{x}_{-i}) < u_i(B, \bm{x}_{-i}) - \delta_i\):
        \begin{equation*}
            P_{\text{LTE}}(x_i = A | \bm{x}_{-i}) = P_{\text{NE}}(x_i = A | \bm{x}_{-i}). 
        \end{equation*}
    \end{enumerate}
\end{enumerate}

Table \ref{tab:comp-result-no-sw-gap} summarizes the computation results of $P(x_i = A | \bm{x}_{-i})$ for a player with 4 neighbors, $k=4$, $k_A\in \{0,1,2,3,4\}$, $\beta=1$ and various values of $\delta_i$ using the payoff matrix in Table \ref{tab:payoff-no-sw-gap} under NE and LTE diffusion models. Under the NE diffusion model, the player will select $A$ with a large probability only if three or more neighbors are in $A$ because the utility of $A$ surpasses the utility of $B$ when $k_A=3$. By contrast, under the LTE model, with a large trust limit, the player will select $A$ with a high probability when only two neighbors are in $A$. By trusting its neighbors and choosing a product that benefits them, the player will benefit in the same manner when it is its neighbor's turn to update the state. The superior alternative spreads faster under the LTE diffusion model than under the NE model, as the long-term return provides an additional incentive.

\begin{table}[h]
    \captionsetup{font=small}
    \caption{Examples of $P(x_i = A | \bm{x}_{-i})$ under NE and LTE models with payoff matrices in Table \ref{tab:payoff-examples}.}
    \begin{subtable}[h]{.225\textwidth}
    \caption{Computation results of $P(x_i = A | \bm{x}_{-i})$ with payoff matrix in Table \ref{tab:payoff-no-sw-gap} under NE and LTE models. $k=4$, $\beta=\beta'=1$.}
    \centering
    \begin{tabular}{|c|c|c|c|}
    \hline
    \multirow{3}{*}{$k_A$} & \multicolumn{3}{c|}{$P(x_i = A | \bm{x}_{-i})$} \\ \cline{2-4}
    & \multirow{2}{*}{NE} & \multicolumn{2}{c|}{LTE} \\ \cline{3-4}
    &  & $\delta_i=1$ & $\delta_i=4$ \\ \hline
    0 & .00034 & .00034 & .00034 \\ \hline
    1 & .00669 & .00669 & .00669 \\ \hline
    2 & .11920 & .11920 & $\approx$ 1 \\ \hline
    3 & .73106 & 1 & 1 \\ \hline
    4 & .98201 & 1 & 1 \\ \hline
    \end{tabular}
    \label{tab:comp-result-no-sw-gap}
    \end{subtable}
    \hspace{+0.2cm}
    \begin{subtable}[h]{.225\textwidth}
    \caption{Computation results of $P(x_i = A | \bm{x}_{-i})$ with payoff matrix in Table \ref{tab:payoff-sw-gap} under NE and LTE models. $k=4$, $\beta=1$.}
    \centering
    \begin{tabular}{|c|c|c|c|}
    \hline
    \multirow{3}{*}{$k_A$} & \multicolumn{3}{c|}{$P(x_i = A | \bm{x}_{-i})$} \\ \cline{2-4}
    & \multirow{2}{*}{NE} &  \multicolumn{2}{c|}{LTE} \\\cline{3-4}
    &  & $\delta_i=1$ & $\delta_i=4$ \\ \hline
    0 & .00034 & 0 & 0 \\ \hline
    1 & .04743 & .04743 & $\approx$ 1 \\ \hline
    2 & .88080 & 1 & 1 \\ \hline
    3 & .99909 & 1 & 1 \\ \hline
    4 & .99999 & 1 & 1 \\ \hline
    \end{tabular}
    \label{tab:comp-result-sw-gap}
    \end{subtable}
\end{table}

From the efficiency perspective, under certain circumstances, the LTE model can push the diffusion process to converge to the social welfare-maximizing equilibrium. Since the NE diffusion process only converges to the risk-dominant equilibrium \cite{harsanyi1988,kandori1993,young2006diffusion}, this equilibrium is not necessarily the most efficient equilibrium that maximizes every player's utility. Risk-dominance depends on the relative magnitude of $a-d$ and $b-c$ rather than $a$ and $b$.  However, an LTE diffusion process can converge to the most efficient equilibrium different from the NE diffusion process if the risk-dominant equilibrium and the social-welfare-maximizing equilibrium differ. Simulation results in Section \ref{sec:simulation} show that when players have moderate to large values of trust limit, the diffusion process under LTE dynamics can converge to the equilibrium where the superior alternative dominates the social network rather than the risk-averse alternative. 

{\bf $\diamond$ Case 2:} Contrary to the payoff matrices without social welfare gap, the payoff matrices in which $c+d \leq 2b$ are termed payoff matrices \emph{with social welfare gap}. An example of such payoff matrices is shown in Table \ref{tab:payoff-sw-gap}. For such payoff matrices, it is no longer guaranteed that the probability that an agent chooses $A$ under the LTE diffusion model is greater than or equal to that under the NE. Under certain circumstances, when there are few players who adopt the superior alternative at the beginning of the diffusion, the superior alternative may rapidly die out in the system because purchasing $A$ results in a decrease of utility to the player as well as its neighbors. The player will not adopt $A$ under such circumstances, as per Case 1 in the LTE diffusion model. Therefore, a sufficient number of initial adopters must be ensured in the system. These initial users are usually the targets under the ``viral marketing'' strategy who are offered free ``samples'' or ``incentives'' to promote the superior alternative by utilizing their social influence \cite{ghayoori2021seed}.

Table \ref{tab:comp-result-sw-gap} summarizes the computation results of $P(x_i = A | \bm{x}_{-i})$ for a player with 4 neighbors, $k=4$, various values of $k_A$ and $\delta_i$, and $\beta=1$ using the payoff matrix in Table \ref{tab:payoff-sw-gap} under NE and LTE diffusion models. Consider a small star network with one central player connected with four leaves. Firstly, the results show that the center player will adopt $A$ with positive probability only if at least one of its neighbors has done so. As an example, suppose that none of its neighbors is in state $A$. Then, product $A$ may spread in the social network only if the center player is selected as the ``seed'' at the beginning of the process. A further calculation for the leaves shows that the system will converge to the all-$A$ state (the state where all players adopt the superior alternative) with positive probability only if at least one player (either center or leaf) is selected as the seed. Without initial adopters, the system will get stuck in the all-$B$ equilibrium state, and product $A$ will never diffuse. Secondly, the center player will adopt $A$ with a probability close to or equal to 1 once one neighbor is in state $A$ if it has a large trust limit and will choose $A$ with a probability 1 once two neighbors are in state $A$ if it behaves trustworthy. This implies that once a sufficient number of individuals have adopted the superior alternative, it will spread much faster under the LTE dynamics than under the NE dynamics because, on top of the individual's instantaneous benefit, the long-term benefit and potential return in the future provide an additional incentive driving the non-adopters to adopt the superior alternative.

\subsection{Numerical Examples for Linear Threshold LTE Diffusion Model}

Consider a player with four neighbors in a star network. Table \ref{tab:comp-result-q-star-1} shows the computation results of $q_w$, $q_u$, and $q^*$ under various values of $\delta'$ using payoff matrix in Table \ref{tab:payoff-no-sw-gap}. Since the payoff matrix is without social welfare gap, $2b-c-d<0$ and $q_w = 0$. $q^*$ decreases from 0.667 to 0 as $\delta'$ increases from 0 to 2; as $\delta'$ further increases, $q^*$ remains 0. Note that the cases of $\delta'=0.25$ and $\delta'=1$ in Table \ref{tab:comp-result-q-star-1} correspond to the cases of $\delta=1$ and $\delta=4$ in Table \ref{tab:comp-result-no-sw-gap}, respectively. The results of Table \ref{tab:comp-result-no-sw-gap} and \ref{tab:comp-result-q-star-1} are consistent. For example, for $\delta' = 1$ ($\delta=4$), when at least two neighbors are in $A$, the player will update its state to $A$. This is in line with the result in Table \ref{tab:payoff-no-sw-gap} that $P(x_i = A \vert \bm{x}_{-i})$ is one or very close to one under the LTE diffusion model when $k_A = 2, 3, 4$. Table \ref{tab:comp-result-q-star-2} shows the computation results using the payoff matrix in Table \ref{tab:payoff-sw-gap}. Since this is a payoff matrix with social welfare gap, $q_w > 0$. $q_w = 0.1$ indicates that only when at least 10\% of the neighbors are in state $A$, choosing $A$ bring about higher social welfare than $B$. As $\delta'$ increases from 0 to 1.5, $q^*$ decreases from 0.4 to 0.1; it does not change with further increase of $\delta'$. The results are again consistent with Table \ref{tab:comp-result-sw-gap}. When $\delta'=0.25$ ($\delta=1$), the player adopts $A$ when at least two neighbors are in $A$ according to the LTE threshold model. In Table \ref{tab:comp-result-sw-gap}, when $k_A = 1$, the player plays according to NE and chooses $A$ with a small probability; when $k_A = 2$, this probability becomes 1.

\begin{table}[h]
    \captionsetup{font=small}
    \caption{Examples of $q^*$ for different payoff matrices.}
    \begin{subtable}[h]{.45\textwidth}
    \caption{Computation results of $q_w$, $q_u$, and $q^*$ for various values of $\delta'$ using payoff matrix in Table \ref{tab:payoff-no-sw-gap}.}
    \centering
    \begin{tabular}{|c|c|c|c|c|c|c|}
    \hline
    $q_w$ & $q_u$ &
    \multicolumn{5}{c|}{$q^*$} \\ \hline
    \multirow{2}{*}{0} & 
    \multirow{2}{*}{0.667} & 
    $\delta'=.25$ & $\delta'=.5$ & $\delta'=1$ & $\delta'=2$ & $\delta'=3$ \\ \cline{3-7}
    & & 0.583 & 0.5 & 0.333 & 0 & 0 \\
    \hline
    \end{tabular}
    \vspace{1em}
    \label{tab:comp-result-q-star-1}
    \end{subtable}
    \begin{subtable}[h]{.45\textwidth}
    \caption{Computation results of $q_w$, $q_u$, and $q^*$ for various values of $\delta'$ using payoff matrix in Table \ref{tab:payoff-no-sw-gap}.}
    \centering
    \begin{tabular}{|c|c|c|c|c|c|c|}
    \hline
    $q_w$ & $q_u$ &
    \multicolumn{5}{c|}{$q^*$} \\ \hline
    \multirow{2}{*}{0.1} & 
    \multirow{2}{*}{0.4} & 
    $\delta'=.25$ & $\delta'=.5$ & $\delta'=1$ & $\delta'=2$ & $\delta'=3$ \\ 
    \cline{3-7}
    & & 0.35 & 0.3 & 0.2 & 0.1 & 0.1 \\ \hline
    \end{tabular}
    \label{tab:comp-result-q-star-2}
    \end{subtable}
\end{table}

We have seen that the range of $q^*$ is between $q_w$ and $q_u$, which is determined by the payoff matrix. For payoff matrices without social welfare gap, $q_w=0$ and $q^*$ can be as small as zero. 
One interesting question is how large $q^*$ can be. For certain ``bad'' payoff matrices, the threshold $q^*$ can be very large. This implies that the player will adopt the superior alternative only if almost all its neighbors have adopted it. Table \ref{tab:payoff-peculiar} shows a peculiar payoff matrix, and Table \ref{tab:comp-result-q-star-peculiar} shows the computation result of $q^*$ under that payoff matrix. When $\delta'$ is small, $q^*$ is large and very close to $q_u$, which is close to 1. In this payoff matrix, $a$ is close to $d$ whereas $b$ is much greater than $c$. This implies that when a few of its neighbors are in state $A$, switching from $B$ to $A$ will cause a significant loss to the player; when the majority of its neighbors are in $A$, changing its state from $B$ to $A$ will not bring much benefit to its utility. Therefore, if the player is more self-interested and has a small value of $\delta'$, the NE dynamics will prevail, and the player will have very little incentive to adopt the superior alternative.

\begin{table}[t]
    \captionsetup{font=small}
    \caption{Examples of $q^*$ for a peculiar payoff matrix.}
    \begin{subtable}[h]{.45\textwidth}
    \caption{A peculiar payoff matrix with large value of $q^*$.}
    \centering
    \begin{tabular}{|c||c|c|}
        \hline
          & $A$ & $B$  \\ \hline \hline
        $A$ & 10, 10 & 1, 9 \\ \hline
        $B$ & 9, 1 & 9, 9 \\ \hline
    \end{tabular}
    \vspace{1em}
    \label{tab:payoff-peculiar}
    \end{subtable}
    \begin{subtable}[h]{.45\textwidth}
    \caption{Computation results of $q_w$, $q_u$, and $q^*$ for various values of $\delta'$ using payoff matrix in Table \ref{tab:payoff-peculiar}}
    \centering
    \begin{tabular}{|c|c|c|c|c|c|}
    \hline
    $q_w$ & $q_u$ &
    \multicolumn{4}{c|}{$q^*$} \\
    \hline
    \multirow{4}{*}{0.444} & 
    \multirow{4}{*}{0.888} & 
    $\delta'=.25$ & $\delta'=.5$ & $\delta'=1$ & $\delta'=2$ \\ 
    \cline{3-6}
    & & 0.861 & 0.833 & 0.778 & 0.667 \\
    \cline{3-6}
    & & $\delta'=3$ & $\delta'=4$ & $\delta'=5$ & $\delta'=6$ \\ \cline{3-6}
    & & 0.556 & 0.444 & 0.444 & 0.444 \\ \hline
    \end{tabular}
    \label{tab:comp-result-q-star-peculiar}
    \end{subtable}
\end{table}

\section{Proof of Theorem \ref{thm:absorbing}} \label{ap:proof-thm-absorbing}

The Markov chain with transition probability matrix $\bm{P}$ characterized by Eq. \eqref{eq:mc-trans-prob-mat} has 2 absorbing states and $(n-1)$ transient states. 
Write $\bm{P}$ in canonical form and partition it into four submatrices: 
\(
    \bm{P} = 
    \begin{bmatrix}
    \bm{T} & \bm{R} \\
    \bm{0} & \bm{I}
    \end{bmatrix},
\)
where $\bm{T}$ and $\bm{R}$ represent the transition among transients states and from transient states to absorbing states, respectively. 
To prove Theorem \ref{thm:absorbing}, it is sufficient to characterize $(\bm{I}-\bm{T})^{-1}$.

\begin{lemma}
The determinant of $\bm{I}-\bm{T}$ is $\Delta$.
\end{lemma}
\begin{proof}
We use superscript to denote the number of players in the social network represented by the 1-dimensional Markov chain. $\bm{P}^{(n)}\in \mb{R}^{(n+1)\times (n+1)}$ denotes the transition probability matrix for a 1-dimensional Markov chain with $n$ players in the social network. We prove this by induction. For the base step, we have $\text{det}[(\bm{I}-\bm{T})^{(2)}] = a_1 + b_1$, and $\text{det}[(\bm{I}-\bm{T})^{(3)}] = (a_1+b_1)(a_2+b_2)-a_1b_2 = a_1 a_2 + a_2 b_1 + b_1 b_2$. Now, for the induction step, we can write 
\begin{align*}
    & \text{det}[(\bm{I}-\bm{T})^{(n+1)}] = 
    \begin{vmatrix}
    (\bm{I}-\bm{T})^{(n)} & 
    \begin{matrix}
    \bm{0} \\ -a_{n-1}
    \end{matrix} \\
    \begin{matrix}
    \bm{0} & -b_n
    \end{matrix} & a_n + b_n
    \end{vmatrix} \\
    & = (a_n + b_n)\text{det}[(\bm{I}-\bm{T})^{(n)}] - a_{n-1} b_n \text{det}[(\bm{I}-\bm{T})^{(n-1)}] \\
    & = \sum_{i=1}^n \prod_{j=i}^n a_j \prod_{j=1}^{i-1}b_j + b_n \sum_{i=1}^n \prod_{j=i}^{n-1} a_j \prod_{j=1}^{i-1} b_j \\
    & \quad -  b_n \sum_{i=1}^{n-1} \prod_{j=i}^{n-1} a_j \prod_{j=1}^{i-1} b_j \\
    & = \sum_{i=1}^n \prod_{j=i}^n a_j \prod_{j=1}^{i-1}b_j + \prod_{j=1}^n b_j \\
    & = \sum_{i=1}^{n+1} \prod_{j=i}^n a_j \prod_{j=1}^{i-1}b_j.
\end{align*}
\end{proof}

\begin{lemma}
The $(i,j)$-th entry of $(\bm{I}-\bm{T})^{-1}$ equals
\begin{align}
    (I-T)^{-1}_{i,j} & = 
    \frac{1}{\Delta} \left( \sum_{k=1}^{\min\{i,j\}} \prod_{l=k}^{j-1} a_l \prod_{l=1}^{k-1} b_l \right) \cdot \notag \\
    & \quad\qquad \left( \sum_{k=\max\{i,j\}+1}^n \prod_{l=k}^{n-1} a_l \prod_{l=j+1}^{k-1} b_l \right). \label{eq:lemma-mat-inv}
\end{align}
\end{lemma}
\begin{proof}
Let $\bm{\alpha}^i$ denote the $i$-th row of $(\bm{I}-\bm{T})^{-1}$ and $\bm{\beta}_{j}$ denote the $j$-th column of $(\bm{I}-\bm{T})$. 
When $i=j$, we have 
\begin{align*}
    \bm{\alpha}^i\cdot \bm{\beta}_j & = -\frac{a_{i-1}}{\Delta} \left( \sum_{k=1}^{i-1} \prod_{l=k}^{i-2} a_l \prod_{l=1}^{k-1} b_l \right) \left( \sum_{k=i+1}^n \prod_{l=k}^{n-1} a_l \prod_{l=i}^{k-1} b_l \right) \\
    & + \frac{a_i + b_i}{\Delta} \left( \sum_{k=1}^{i} \prod_{l=k}^{i-1} a_l \prod_{l=1}^{k-1} b_l \right) \left( \sum_{k=i+1}^n \prod_{l=k}^{n-1} a_l \prod_{l=i+1}^{k-1} b_l \right) \\
    & - \frac{b_{i+1}}{\Delta} \left( \sum_{k=1}^{i} \prod_{l=k}^{i} a_l \prod_{l=1}^{k-1} b_l \right) \left( \sum_{k=i+2}^n \prod_{l=k}^{n-1} a_l \prod_{l=i+2}^{k-1} b_l \right) \\
    & = \frac{1}{\Delta} \left( \sum_{k=1}^{i-1} \prod_{l=k}^{n-1} a_l \prod_{l=1}^{k-1} b_l \right) + \frac{1}{\Delta} \left( \sum_{k=i+2}^n \prod_{l=k}^{n-1} a_l \prod_{l=1}^{k-1} b_l \right) \\
    & + \frac{a_i+b_i}{\Delta} \prod_{l=i+1}^{n-1}a_l \prod_{l=1}^{i-1} b_l = \frac{1}{\Delta}\left(\sum_{k=1}^{n} \prod_{l=k}^{n-1} a_l \prod_{l=1}^{k-1} b_l\right) \\
    & = 1.
\end{align*}
When $i\neq j$, without loss of generality, let $i<j$. We have 
\begin{align*}
    \bm{\alpha}^i\cdot \bm{\beta}_j & = -\frac{a_{j-1}}{\Delta} \left( \sum_{k=1}^{i} \prod_{l=k}^{j-2} a_l \prod_{l=1}^{k-1} b_l \right)
    \left( \sum_{k=j}^n \prod_{l=k}^{n-1} a_l \prod_{l=j}^{k-1} b_l \right) \\
    & + \frac{a_j + b_j}{\Delta} \left( \sum_{k=1}^{i} \prod_{l=k}^{j-1} a_l \prod_{l=1}^{k-1} b_l \right) \left( \sum_{k=j+1}^n \prod_{l=k}^{n-1} a_l \prod_{l=j+1}^{k-1} b_l \right) \\
    & - \frac{b_{j+1}}{\Delta} \left( \sum_{k=1}^{i} \prod_{l=k}^{j} a_l \prod_{l=1}^{k-1} b_l \right) \left( \sum_{k=j+2}^n \prod_{l=k}^{n-1} a_l \prod_{l=j+2}^{k-1} b_l \right) \\ 
    & = 0.
\end{align*}
Since $(\bm{I}-\bm{T})$ is invertible, Eq. \eqref{eq:lemma-mat-inv} characterizes its inverse. 
\end{proof}

Denote the absorption probability and expected time to absorption by $\bm{W}$ and $\bm{\tau}$ respectively. 
The theorem follows as $\bm{W} = (\bm{I}-\bm{T})^{-1} \bm{R}$ and $\bm{\tau}=(\bm{I}-\bm{T})^{-1}\bm{1}$ where $\bm{1}$ denotes the one-vector.

\section{Additional Simulation Results}

In this section, we first numerically demonstrate the effect of parameters $\beta, \beta'$. Next, we compare the numerical results of the best-response LTE diffusion model and linear threshold diffusion model and show their equivalence.

\subsection{Effect of Sensitivity Parameters $\beta, \beta'$}

\begin{figure}[ht]
    \captionsetup{font=small}
    \centering
    \begin{subfigure}[h]{0.48\textwidth}
    \centering
    \includegraphics[trim=0.9in 0.25in 0.9in 0.25in,clip,width=\textwidth]{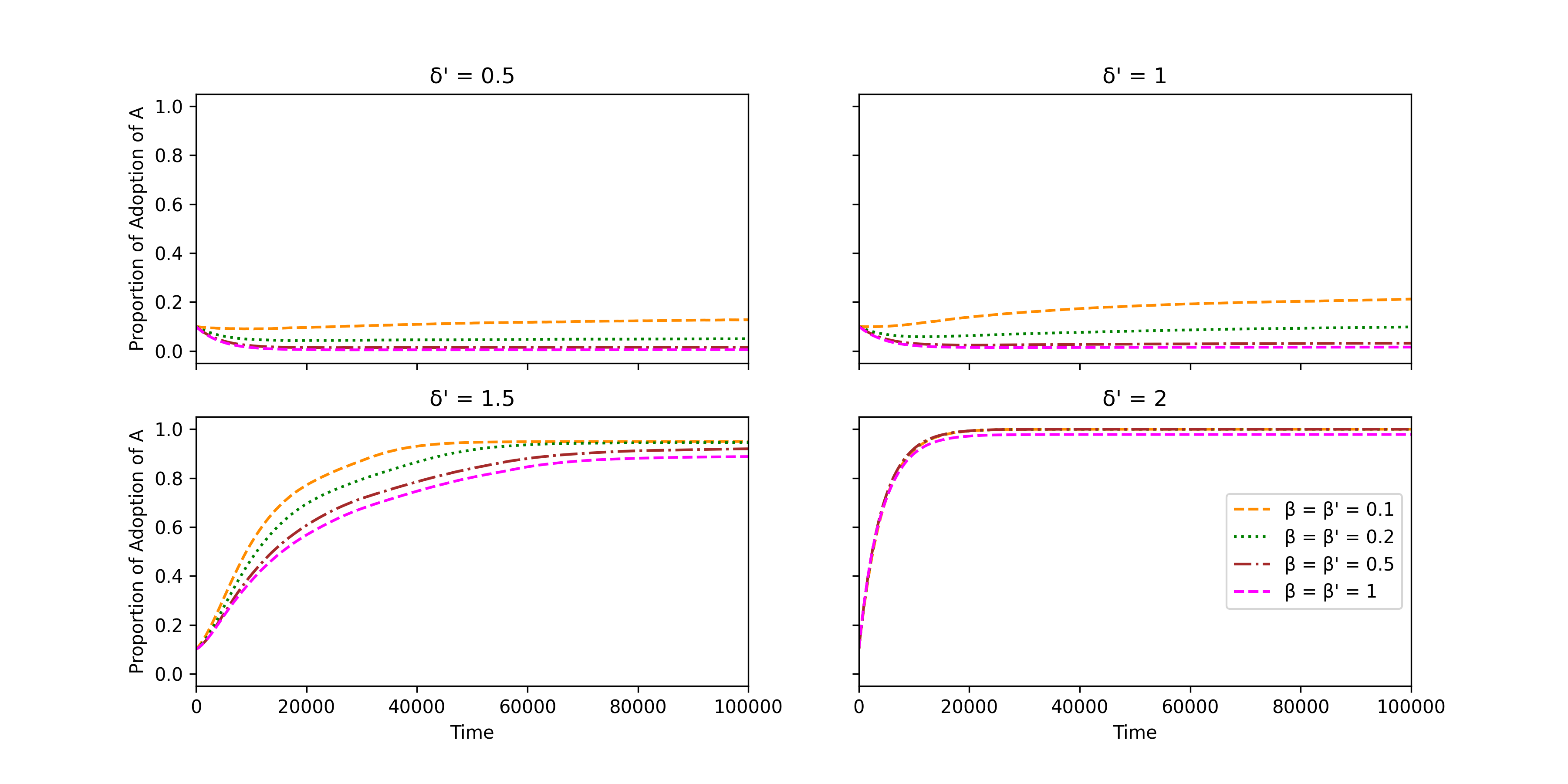}
    \subcaption{Initial adopter = 10\% of the total number of players.}
    \label{subfig:sim-egofb-beta-eff-10}
    \end{subfigure}
    \begin{subfigure}[h]{0.48\textwidth}
    \centering
    \includegraphics[trim=0.9in 0.25in 0.9in 0.25in,clip,width=\textwidth]{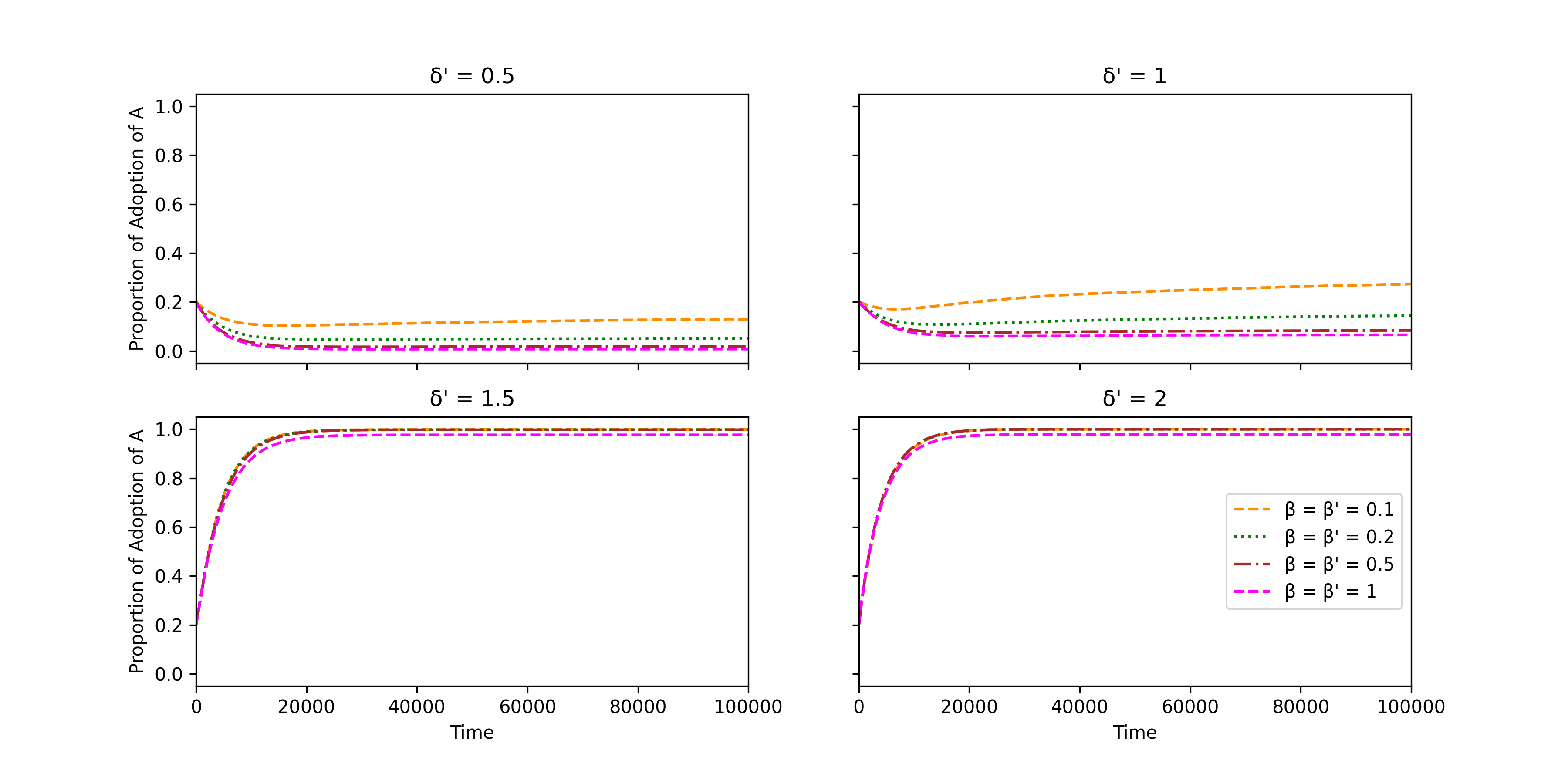}
    \subcaption{Initial adopter = 20\% of the total number of players.}
    \label{subfig:sim-egofb-beta-eff-20}
    \end{subfigure}
    \caption{Simulation results of diffusion on ego-Facebook social network with different values of $\beta$ and $\beta'$.}
    \label{fig:sim-egofb-beta-eff}
\end{figure}

In this subsection, we study the effect of sensitivity parameters $\beta$ and $\beta'$ numerically via simulation. We study the effect of the magnitude of $\beta$ and $\beta'$ on the diffusion process over the ego-Facebook social network. The same setting with Section \ref{subsec:diff-sim} is adopted except that values of $\beta$ and $\beta'$ are varied. We consider various values of $\beta$ and $\beta'$ from 0.1 to 1, and we set $\beta = \beta'$ in all cases. The payoff matrix in Table \ref{tab:payoff-sim-1} is adopted. Various values of $\delta'$ between 0.5 and 2 are considered, and the results for 10\% and 20\% initial adopters are presented in Fig. \ref{fig:sim-egofb-beta-eff}. In each plot are plotted the diffusion outcomes with different values of $\beta$ and $\beta'$ but the same $\delta'$ and number of initial adopters.

The results in Fig. \ref{fig:sim-egofb-beta-eff} show that different values of $\beta$ and $\beta'$ do not affect the general trend and direction of diffusion evolution. The diffusion process under different values of $\beta$ and $\beta'$ converge in the same direction. However, the final equilibrium state may differ slightly. In the cases shown in Fig. \ref{fig:sim-egofb-beta-eff}, the diffusion process converges to an equilibrium state with a larger proportion of A in the social network when the values of $\beta$ and $\beta'$ are smaller, regardless of the direction that it converges to. When the diffusion process starts from a state where A is sparse in the social network and finally converges to a state where A dominates the social network, the superior alternative A spreads over the network faster when $\beta$ and $\beta'$ are smaller.

\subsection{Comparison between BR-LTE and LT-LTE Models}

\begin{figure}[ht]
    \captionsetup{font=small}
    \centering
    \begin{subfigure}[h]{0.48\textwidth}
    \centering
    \includegraphics[trim=0.9in 0.25in 0.9in 0.25in,clip,width=\textwidth]{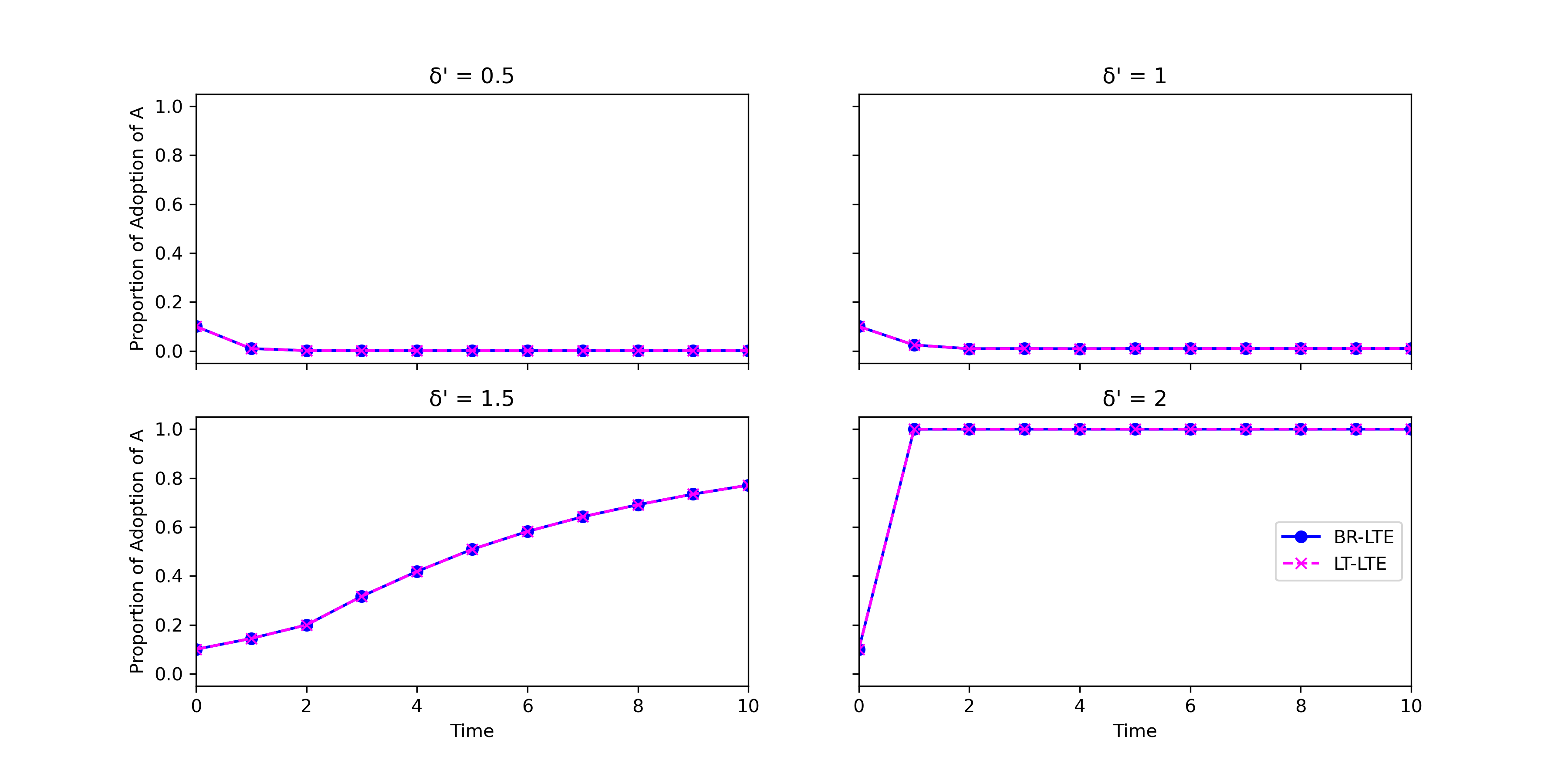}
    \subcaption{Initial adopter = 10\% of the total number of players.}
    \label{subfig:sim-egofb-br-lt-10}
    \end{subfigure}
    \begin{subfigure}[h]{0.48\textwidth}
    \centering
    \includegraphics[trim=0.9in 0.25in 0.9in 0.25in,clip,width=\textwidth]{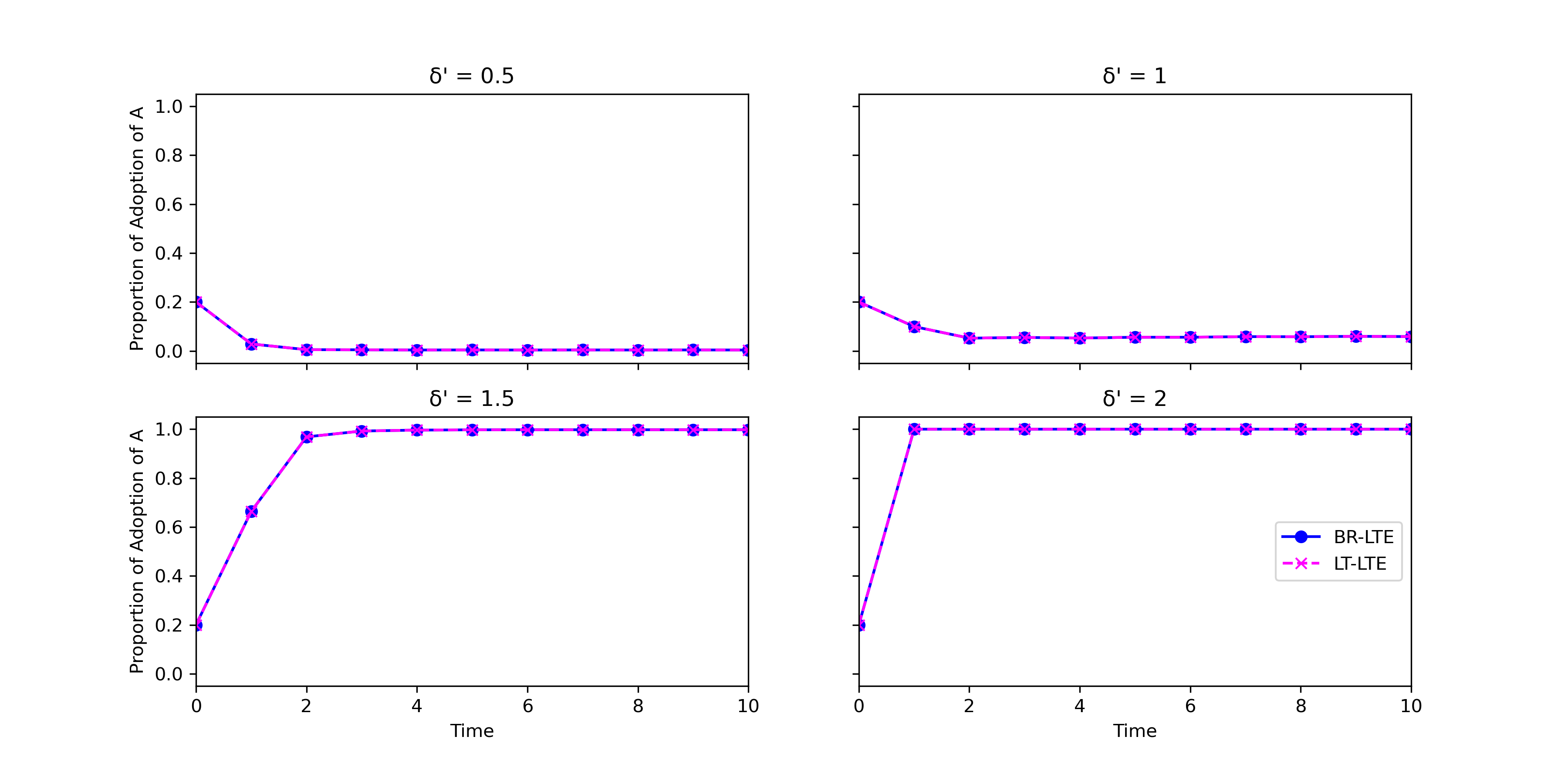}
    \subcaption{Initial adopter = 20\% of the total number of players.}
    \label{subfig:sim-egofb-br-lt-20}
    \end{subfigure}
    \caption{Simulation results of synchronous BR-LTE and LT-LTE models on ego-Facebook social network.}
    \label{fig:sim-egofb-br-lt}
\end{figure}

In Section \ref{sec:threshold-model}, we introduced the best-response LTE (BR-LTE) diffusion model in which the randomness of the LTE diffusion model is eliminated. This is equivalent to the LTE diffusion model by setting $\beta \to \infty$ and $\beta' \to \infty$. We also established the equivalence between BR-LTE model and the well-known linear threshold (LT) diffusion model when all players update their purchase option synchronously at each time instant. Such LT diffusion model in the context of LTE is called the linear threshold LTE model and abbreviated as LT-LTE model. In the last subsection of simulation, we present the simulation results of the synchronous BR-LTE and LT-LTE models over the ego-Facebook social network under the same settings as Section \ref{subsec:diff-sim}. The payoff matrix in Table \ref{tab:payoff-sim-1} is adopted, and various values of $\delta'$ from 0.5 to 2 are considered. The simulation results for 10\% and 20\% initial adopters are presented in Fig. \ref{fig:sim-egofb-br-lt}. In each plot, there are two curves: the blue one for the synchronous BR-LTE model and the magenta one for LT-LTE model. Recall that both models are synchronous, and therefore, the diffusion process converges much faster than the asynchronous model presented in the previous subsections. The curves for synchronous BR-LTE and LT-LTE models coincide and yield the same result. Therefore, the simulation results numerically demonstrate the equivalence of synchronous BR-LTE and LT-LTE models under their own diffusion dynamics and mechanisms.


\end{document}